\documentclass[a4paper,american]{article}

\usepackage[utf8]{inputenc}
\usepackage[lighttt]{lmodern}
\usepackage{microtype}
\usepackage{hyperref}

\usepackage[autopunct]{csquotes}
\MakeOuterQuote{"}

\usepackage{enumitem}
\setlist{noitemsep}

\usepackage{amsmath,amssymb,amsthm}
\usepackage[bbgreekl]{mathbbol}
\allowdisplaybreaks

\usepackage{thmtools,thm-patch}
\theoremstyle{plain}
\declaretheorem{theorem}
\declaretheorem[sibling=theorem]{proposition}
\declaretheorem[sibling=theorem]{lemma}
\declaretheorem[sibling=theorem]{corollary}

\theoremstyle{definition}
\declaretheorem{definition}

\usepackage{float}
\floatstyle{boxed}
\newfloat{algfloat}{tb}{alg}
\declaretheorem[preheadhook={\begin{algfloat}},postfoothook={\end{algfloat}}]{algorithm}
\declaretheorem[title=API,sibling=algorithm,preheadhook={\begin{algfloat}},postfoothook={\end{algfloat}}]{api}

\usepackage{listings}
\usepackage{lstautogobble}
\usepackage{tikz}
\usetikzlibrary{arrows.meta}
\usetikzlibrary{automata}
\usetikzlibrary{decorations.pathmorphing}
\makeatletter
\global\let\tikz@ensure@dollar@catcode=\relax
\makeatother

\usepackage{pgfplots}
\usetikzlibrary{plotmarks}
\pgfplotsset{compat=1.14}

\usepackage{subcaption}
\usepackage{tabto}
\usepackage[backgroundcolor=yellow]{todonotes}

\usepackage{authblk}
\author[1]{Grzegorz Herman}
\affil[1]{\footnotesize
  Theoretical Computer Science, Jagiellonian University
  \newline
  \texttt{gherman@tcs.uj.edu.pl}
}

\title{Relational parsing: a clean, fast parsing strategy for all context-free languages
\thanks{This manuscript is a work in progress:
a few proofs still need to be written down,
and there are many rough edges.
We also hope to provide more experimental results and extend the discussion in the final section.
Thus, please contact us first if you would like to cite this work.
We also welcome comments and suggestions of improvements,
regarding both the algorithm itself and the way it is presented here.}}

\begin{document}
\TabPositions{.55\columnwidth}

\newcommand{\bld}[1]{{\boldmath\textbf{#1}\unboldmath}}
\newcommand{\sstart}{s_{\textup{start}}}
\newcommand{\sstop}{s_{\textup{stop}}}
\newcommand{\con}{{\cdot}}
\newcommand{\cons}{{\cdots}}
\newcommand{\deq}{\triangleq}
\newcommand{\assign}{\mathrel{\mathop:}=}
\newcommand{\shift}[1]{\textup{\texttt{shift}}\langle#1\rangle}
\newcommand{\call}[1]{\textup{\texttt{call}}\langle#1\rangle}
\newcommand{\reduce}[1]{\textup{\texttt{reduce}}\langle#1\rangle}
\newcommand{\rev}[1]{{%
    \begin{tikzpicture}[baseline={(tmp.base)}]
    \node[inner sep=0pt,outer sep=0pt] (tmp) {$#1$};
    \draw[-{To[width=.8ex]}] ([yshift=.5ex]tmp.north east) -- ([yshift=.5ex]tmp.north west);
    \end{tikzpicture}
  }}
\newcommand{\move}[1][~]{%
  \mathrel{%
    \begin{tikzpicture}[baseline={(base)}]
    \node[inner sep=.5ex,align=center] (tmp) {$\scriptstyle #1$};
    \node (base) at ([yshift=-.6ex]tmp.south) {};
    \draw[-{To[width=.8ex]}] decorate [decoration={coil,amplitude=.1ex,segment length=.5ex,aspect=0,post length=.4ex}] {(tmp.south west) -- (tmp.south east)};
    \end{tikzpicture}
    }
  }
\newcommand{\edge}[1][~]{%
  \mathrel{%
    \begin{tikzpicture}[baseline={(base)}]
    \node[inner sep=.5ex,align=center] (tmp) {$\scriptstyle #1$};
    \node (base) at ([yshift=-.6ex]tmp.south) {};
    \draw[-{To[width=.8ex]}] (tmp.south west) -- (tmp.south east);
    \end{tikzpicture}
    }
  }
\newcommand{\Eps}{D_\epsilon^*}
\newcommand{\C}[1]{\mathcal{C}(#1)}
\renewcommand{\S}[2][]{{\mathcal{S}_{#1}(#2)}}
\newcommand{\N}[1]{\overline{#1}}
\newcommand{\D}[2]{\Delta_{#1 \rightsquigarrow #2}}
\renewcommand{\P}[2]{\Phi_{#1 \rightsquigarrow #2}}
\renewcommand{\L}[2]{\Lambda_{#1 \rightsquigarrow #2}}
\newcommand{\nulling}[1]{\D{#1}\epsilon}
\newcommand{\red}{_{\scriptscriptstyle-}}
\newcommand{\nred}{_{\scriptscriptstyle+}}
\newcommand{\deriv}[2][]{^{\textup{#1}(#2)}}
\newcommand{\efree}[1]{\underline{#1}}
\renewcommand{\flat}[1]{\left[#1\right]}
\newcommand{\type}[1]{\mathbb{#1}}
\newcommand{\out}[1][out]{_{\textup{#1}}}
\newcommand{\val}[1]{\langle #1 \rangle}
\newcommand{\dep}{\succ}

\maketitle

\begin{abstract}
We present a novel parsing algorithm for all context-free languages,
based on computing the relation between configurations and reaching transitions in a recursive transition network.
Parsing complexity w.r.t. input length matches the state of the art:
it is worst-case cubic, quadratic for unambiguous grammars, and linear for LR-regular ones.
What distinguishes our algorithm is its clean mathematical formulation:
parsing is expressed as a composition of simple operations on languages and relations,
and can therefore be implemented using only immutable data structures.
With a proper choice of these structures,
a vast majority of operations performed during parsing typical programming languages can be memoized,
which allows our proof-of-concept implementation to outperform common generalized parsing algorithms,
in some cases by orders of magnitude.
\end{abstract}

\section{Introduction}

Parsing is a well-studied and yet still active topic in computer science.
Multiple highly efficient deterministic parsers can handle large families of unambiguous context-free grammars, and oftentimes more.
However, grammars of many programming languages do not naturally fit within the limitations required by these parsers, and need tedious manual adaptation
(e.g., replacing left recursion with right recursion).
Worse, such adaptation needs to be "undone" on parse trees, in order to make them reflect the desired structure of the language.
Further, in many cases the "natural" grammar of a language allows some ambiguity, which is easier resolved at a later compilation step.
With the growing need for rapid development of new languages,
generalized parsers---able to handle all, even ambiguous, context-free grammars---are gaining wider adoption.

Most generalized parsing techniques are based on clever parallel simulation of nondeterministic choices of a pushdown automaton.
Some of these algorithms offer good complexity guarantees: worst-case cubic, quadratic on unambiguous grammars, and in some cases linear on "nice" language fragments.
Nondeterminism is handled by means of (some variant of) a data structure called GSS (graph-structured stack).
This makes the constant factors involved significantly larger than in deterministic parsers, which use a simple stack.
Moreover, in order to handle the full class of context-free grammars,
the graph-like data structures in \emph{all} these algorithms need to be \emph{mutable} and allow \emph{cycles}.
Even where it does not significantly complicate the implementation,
it makes reasoning about the algorithm more difficult.

We present a new generalized context-free parsing algorithm which we call \bld{relational parsing}.
Similarly to existing approaches, it simulates possible runs of a pushdown automaton (specifically, a recursive transition network).
However, its formulation is based on precise mathematical semantics, given in Section~\ref{sec:recognition}:
it inductively computes the languages of configurations reachable after reading consecutive input symbols.
We show that each such language can be obtained from the previous one by a constant number of simple, well-known language-theoretic operations.
The way of representing these languages (and their grammar-dependent atomic building blocks) is left abstract at this point.

Our recognition algorithm generalizes cleanly to a full-fledged parser, capable of producing a compact representation of all possible parse trees.
To achieve that, instead of languages of reachable configurations, we compute the relations between configurations and the ways in which they can be reached
(hence the algorithm name).
This generalization is presented in Section~\ref{sec:parsing}.
Again, the representation of relations is hidden behind an abstract API.

We give possible concrete realizations of all required abstract data types in Section~\ref{sec:representations}.
For the relations computed during parsing, we use a DAG structure akin to the graph-structured stack.
However, with our approach, all cyclicity is embedded in the atomic relations,
and thus we \emph{never} need to add edges to existing DAG nodes.
The atomic languages and relations are shown to be regular,
and thus can be represented by nondeterministic finite automata.
We also discuss some possible representations of transition languages.

Instead of computing the language of all accepting transition sequences,
our parsing algorithm can work with any semiring induced by individual transitions.
However, as presented, it introduces artificial ambiguity into the objects computed,
and thus requires such semiring to be idempotent.
In Section~\ref{sec:semiring} we pinpoint the reasons for this ambiguity,
and show how it can be eliminated by a simple yet subtle adjustment of the atomic relations.

The input-related complexity of our algorithm, discussed in Section~\ref{sec:complexity}, matches the state of the art:
our parser is worst-case cubic, quadratic on unambiguous grammars, and linear on LR-regular grammars.

There are multiple benefits of our clean, semantic approach.
The automata representing atomic languages can be optimized.
More importantly, due to immutability, partial results can be memoized and reused.
In Section~\ref{sec:memoization} we outline how to achieve such reuse by using a more complex representation of languages.
On real-world inputs (Java~8 grammar and a large collection of source files)
it allows a \emph{vast majority} of phases to be performed on a plain stack,
which makes our (unoptimized) proof-of-concept C$^\sharp$ implementation significantly outperform (our, unoptimized implementation of) GLL.
In recognition mode the speed-up is close to \emph{two orders of magnitude}---%
the running times are on par with those of the hand-coded deterministic parser used by the Java compiler.
This experimental evaluation is presented in Section~\ref{sec:evaluation}.

Finally, in Section~\ref{sec:related} we discuss how our approach fits into the generalized parsing landscape.
We show how the bits and pieces of our algorithm relate to the seminal work of Lang, Earley, and Tomita,
as well as to GLL, GLR, GLC (generalized left-corner), reduction-incorporated parsers, and Adaptive LL(*) employed by ANTLR4.
Within the discussion we also hint at some places where our algorithm could be improved, or our approach extended.

Our proof-of-concept implementation can be accessed at\\
\url{https://gitlab.com/alt-lang/research/relational-parsing}.

\section{Recognition}
\label{sec:recognition}

There are many equivalent formalisms capturing context-free languages.
In this work, we use \emph{recursive transition networks}~\cite{rtn},
a restricted class of nondeterministic pushdown automata.
They directly model the computational behavior of (nondeterministic) recursive descent parsers,
and at the same time can be \emph{trivially} obtained from context-free grammars.%
\footnote{This includes popular variants of context-free grammars in which regular expressions can be used at right-hand sides of production rules.}
Importantly, the states obtained in such translation are exactly the "items" used by GLL, (G)LC, and (G)LR parsers.
Let us begin with a quick review of the chosen model.

\begin{definition}
  A \bld{recursive transition network} over a terminal alphabet $A$ and set of production labels%
  \footnote{The labels are never looked at by our algorithm.
  However, they can be attached to internal parse tree nodes---%
  a role usually served by grammar productions, hence the name.}
  $P$ consists of:
  \begin{itemize}
    \item a finite set $S$ of states, with a distinguished start state $\sstart \in S$,
    \item a finite set $D$ of transitions, each being one of:
      \begin{itemize}
        \item $\shift{s,a,t}$ with $s,t \in S$ and $a \in A$,
        \item $\call{s,u,t}$ with $s,u,t \in S$, or
        \item $\reduce{s,p}$ with $s \in S$ and $p \in P$.
      \end{itemize}
  \end{itemize}
\end{definition}

Throughout the paper, we assume to be working with a single network,
and thus consider $A$, $S$, $D$ and $P$ to be fixed.
For convenience, we let $D_a$ denote the set of shifts for each particular symbol $a \in A$,
and $D_\epsilon$---the set of all non-shift transitions.

Recognition by an RTN begins in the state $\sstart$.
At each moment, the network may nondeterministically choose any legal transition from its current state $s$.
A $\shift{s,a,t}$ consumes the symbol $a$ from the input and changes state to $t$,
$\call{s,u,t}$ recursively invokes recognition from state $u$ and moves to $t$ once the recursive call returns,
and $\reduce{s,p}$ returns from a single level of recursion.
An input $\alpha$ belongs to the language of the RTN
if there is a sequence of legal transitions which consumes the whole $\alpha$,
and ends by returning from the top recursion level.
If a parse tree is to be constructed,
its shape follows exactly the nesting of recursive calls made,
with $p$ from $\reduce{s,p}$ used to label the resulting node.

Making explicit the recursion stack, the recognition process can be formalized as follows:

\begin{definition}
  A \bld{configuration} of a recursive transition network is a stack of its states, i.e. a member of $S^*$
  (the top---current---state is the first in such sequence).
  Each transition $d \in D$ induces the following \bld{reachability} relation $\move[d]$ between configurations:
  $$
    s \con \sigma \move[\shift{s,a,t}] t \con \sigma, \qquad
    s \con \sigma \move[\call{s,u,t}] u \con t \con \sigma, \qquad
    s \con \sigma \move[\reduce{s,p}] \sigma.
  $$
  Reachability generalizes to sequences of transitions (by composing relations for consecutive transitions),
  and then to languages of transitions (by taking union of relations for individual transition sequences).
  A sequence of transitions is \bld{viable} if it induces a nonempty reachability relation.
\end{definition}

Our recognition algorithm will simulate the nondeterministic choices made by the RTN
by computing the sets of configurations reachable after consuming consecutive symbols from the input.
To avoid dealing with empty configurations, we augment the set of states $S$ by a fresh state $\sstop$ (with no outgoing transitions)
and use it as a guard marking the bottom of each configuration.

\begin{definition}
  For an input $\alpha \in A^*$,
  the \bld{parsing language} $\Sigma_{\alpha} \subseteq S^*$
  consists of configurations (augmented with $\sstop$) reachable from the initial configuration when parsing $\alpha$:
  $$
    \Sigma_{a_1 \cons a_n} \deq \{\sigma \con \sstop:\ \sstart \move[\Eps \con D_{ a_1} \cons \Eps \con D_{a_n} \con \Eps] \sigma \}.
  $$
\end{definition}

The question of whether an input belongs to the language recognized by the network
is now the same as whether its parsing language contains the configuration $\sstop$.
The following definition allows constructing this language inductively, symbol by symbol:

\begin{definition}
  The \bld{call-reduce closure} and \bld{shift by a symbol $a \in A$} of a language $\Sigma \subseteq S^*$
  are the sets of configurations reachable from $\Sigma$ without shifts and by shifting $a$, respectively:
  \begin{align*}
    \C\Sigma &\deq \{\sigma':\ \sigma \move[\Eps]\sigma',\ \sigma\in\Sigma \}, \\
    \S[a]\Sigma &\deq \{\sigma':\ \sigma \move[D_a] \sigma',\ \sigma\in\Sigma \}.
  \end{align*}
\end{definition}

Based on the obvious observation that $\Sigma_{\epsilon} = \C\sstart \con \sstop = \C\sstart \con \C\sstop$ and $\Sigma_{\alpha a} = \C{\S[a]{\Sigma_{\alpha}}}$,
we propose Algorithm~\ref{alg:main-recognize} as a generic blueprint for context-free recognition.
Assume for the moment that we already have the closures of individual states.
The main line of development of this paper is showing how line~3, called a \bld{phase}, can be performed efficiently.

\begin{algorithm}\label{alg:main-recognize}
  Checking whether $\alpha \in A^*$ belongs to the language of the RTN
  \begin{lstlisting}
  $\Sigma \assign \C\sstart \con \C\sstop$
  foreach $a \in \alpha$ (in order):
    $\Sigma \assign \expandafter\C{\S[a]\Sigma}$
  return whether $\sstop \in \Sigma$
  \end{lstlisting}
\end{algorithm}

\begin{algorithm}\label{alg:phase-recognize-naive}
  Calculating $\C{\S[a]\Sigma}$
  \begin{lstlisting}
  $\Sigma\out \assign \varnothing$
  foreach $s_0 s_1 \cons s_k \in \Sigma$:
    foreach $\shift{s_0,a,t} \in D_a$:
      $\Sigma\out \assign \Sigma\out \ \cup\  \C{t} \con s_1 \cons s_k$
      if $t \move[\Eps]\epsilon$:
        $\Sigma\out \assign \Sigma\out \ \cup\  \C{s_1} \con s_2 \cons s_k$
        if $s_1 \move[\Eps]\epsilon$:
          $\Sigma\out \assign \Sigma\out \ \cup\  \C{s_2} \con s_3 \cons s_k$
          ...
  return $\Sigma\out$
  \end{lstlisting}
\end{algorithm}

Let us begin with a na\"ive implementation, presented as Algorithm~\ref{alg:phase-recognize-naive}.
While obviously correct, it is not directly usable,
as the loop in line~2 ranges over potentially infinite $\Sigma$.
To deal with this issue, we rephrase it in terms of language derivatives.

\begin{definition}
  A \bld{derivative of degree $k$} of a language $\Sigma \subseteq S^*$ by a symbol $s \in S$
  is any collection of languages $\Sigma\deriv{s} = \{\Sigma_i\}_{i = 1 \ldots k}$
  such that for each $\sigma \in S^*$ we have%
  \footnote{On the right-hand side of the equivalence, $s$ is completely redundant.
  We keep it for consistency with later generalizations of this definition.}
  $$
    s \con \sigma \in \Sigma  \iff  s \con \sigma \in \bigcup_i s \con \Sigma_i.
  $$
  A language is \bld{$(m,k)$-regular} if it belongs to some family of size $m$, closed under derivatives of degree $k$.
  It is \bld{$k$-regular} if it is $(m,k)$-regular for some (finite) $m$, and \bld{regular} if it is $k$-regular for some $k$.
\end{definition}

This definition of regular languages is equivalent to the standard one.
The parameter $m$ measures the size of an automaton recognizing such language,
while $k$---the degree of its nondeterminism.
In particular, derivatives of degree one are exactly those of Brzozowski~\cite{Brzozowski}.

Derivatives let us look at the first (top) state in a configuration,
keeping everything that lies below it encapsulated.
This leads to Algorithm~\ref{alg:phase-recognize-deriv} which, while still operating on potentially infinite objects,
no longer needs to explicitly examine their contents.
As we will show in Section~\ref{sec:representations}, all languages encountered during recognition are regular,
and thus have a finite representation on which the necessary operations can be performed efficiently.

\begin{algorithm}\label{alg:phase-recognize-deriv}
  Calculating $\C{\S[a]\Sigma}$ using derivatives
  \begin{lstlisting}
  $\Sigma\out \assign \varnothing$
  foreach $\shift{s,a,t} \in D_a$:
    foreach $\Sigma_0 \in \Sigma\deriv{s}$:
      $\Sigma\out \assign \Sigma\out \ \cup\  \C{t} \con \Sigma_0$
      if $t \move[\Eps]\epsilon$, foreach $s_1 \in S$:
        foreach $\Sigma_1 \in \Sigma_0\deriv{s_1}$:
          $\Sigma\out \assign \Sigma\out \ \cup\  \C{s_1} \con \Sigma_1$
          if $s_1 \move[\Eps]\epsilon$, foreach $s_2 \in S$:
            foreach $\Sigma_2 \in \Sigma_1\deriv{s_2}$:
              $\Sigma\out \assign \Sigma\out \ \cup\  \C{s_2} \con \Sigma_2$
              ...
  return $\Sigma\out$
  \end{lstlisting}
\end{algorithm}

Nested iteration in Algorithm~\ref{alg:phase-recognize-deriv} could be implemented using bounded recursion.
However, the only reason for the nesting is taking derivatives by successive nullable states
(a state $s$ is \bld{nullable} if $s \move[\Eps] \epsilon$).
Let us therefore introduce a condition on languages, requiring all nullable states to be (optionally) skipped:

\begin{definition}
  A language $\Sigma \subseteq S^*$ is \bld{null-closed}
  if for each $S \ni s \move[\Eps] \epsilon$ and $\sigma_1,\sigma_2 \in S^*$ we have
  $$
    \sigma_1 \con s \con \sigma_2 \in \Sigma \implies \sigma_1\con\sigma_2 \in \Sigma.
  $$
  The \bld{null closure} of a language $\Sigma$ is the smallest null-closed language $\N\Sigma \supseteq \Sigma$.
\end{definition}

Now, for null-closed $\Sigma$, the phase algorithm can be rewritten as Algorithm~\ref{alg:phase-recognize-deep}.
Lines~5--7 are necessary because the state $s'$ could have entered the language "in the middle" of a call-reduce closure,
and thus might itself not have been call-reduce closed.

\begin{algorithm}\label{alg:phase-recognize-deep}
  Calculating $\N{\C{\S[a]\Sigma}}$ for null-closed $\Sigma$
  \begin{lstlisting}
  $\Sigma\out \assign \varnothing$
  foreach $\shift{s,a,t} \in D_a$:
    foreach $\Sigma_0 \in \Sigma\deriv{s}$:
      $\Sigma\out \assign \Sigma\out \ \cup\  \N{\C{t}} \con \Sigma_0$
      if $t \move[\Eps]\epsilon$, foreach $s' \in S$:
        foreach $\Sigma_1 \in \Sigma_0\deriv{s'}$:
          $\Sigma\out \assign \Sigma\out \ \cup\  \N{\C{s'}} \con \Sigma_1$
  return $\Sigma\out$
  \end{lstlisting}
\end{algorithm}

We would now like to take Algorithm~\ref{alg:main-recognize} and replace all call-reduce closures by their null closures.
The correctness of such change needs to be justified---we do so by the following reasoning.

\begin{lemma}[Canonical decomposition]
  Every viable sequence of non-shift transitions $\eta \in \Eps$
  can be uniquely decomposed into two parts:
  \begin{itemize}
    \item \bld{reductive} $\eta_{-}$, which nulls some configuration, i.e.,
      $S^* \ni \sigma\red \move[\eta\red] \epsilon$, and
    \item \bld{nonreductive} $\eta_{+}$, which does not null some configuration consisting of a single state, i.e.,
      $S \ni s\nred \move[\eta\nred] \sigma\nred \in S^+$.
  \end{itemize}
\end{lemma}
\begin{proof}
  The proof is by induction on the length of $\eta$.
  If $\eta = \epsilon$, the decomposition is $\eta\red = \epsilon$ (with $\sigma\red = \epsilon$) and $\eta\nred = \epsilon$ (with $\sigma\nred = s\nred$ being any state in $S$).

  Otherwise, let $s_0$ be the source state of the first transition in $\eta$.
  If $s_0 \move[\eta] \sigma\nred$ for some $\sigma\nred \in S^+$, the decomposition is $\eta\red = \epsilon$ and $\eta\nred = \eta$ (with $s\nred = s_0$).
  Otherwise, let $\eta = \eta_0\con\eta'$ with $\eta_0$ being the longest prefix of $\eta$ valid from $s_0$.
  It must be that $s_0 \move[\eta_0] \epsilon$.
  Use the induction hypothesis to decompose $\eta' = \eta'\red\con\eta'\nred$ (witnessed by $\sigma'\red$, $s'\nred$ and $\sigma'\nred$),
  and set $\eta\red = \eta_0\con\eta'\red$ (with $\sigma\red = s_0\con\sigma'\red$) and $\eta\nred = \eta'\nred$ (with $s\nred = s'\nred$ and $\sigma\nred = \sigma'\nred$).

  Uniqueness follows by the same induction from the fact that no proper prefix of $\eta_0$ nulls $s_0$.
\end{proof}

\begin{lemma}
  \label{lem:shift-recognize}
  $\S[a]{\N\Sigma} \subseteq \N{\S[a]\Sigma}$ for every call-reduce-closed language $\Sigma \subseteq S^*$.
\end{lemma}
\begin{proof}
  Every configuration in $\S[a]{\N\Sigma}$ comes from
  $$
    \N\Sigma \ni s\con\sigma \move[\shift{s,a,t} \in D_a] t\con\sigma
    \quad\implies\quad
    t\con\sigma \in \S[a]{\N\Sigma}.
  $$
  Then there must exist $\sigma_1 \con s \con \sigma_2 \in \Sigma$
  with $\epsilon \in \N{\sigma_1}$ and $\sigma \in \N{\sigma_2}$.
  As $\Sigma$ is call-reduce-closed, we have $s \con \sigma_2 \in \Sigma$,
  from which $t \con \sigma_2 \in \S[a]\Sigma$
  and $t \con \sigma \in \N{\S[a]\Sigma}$ as desired.
\end{proof}

\begin{lemma}
  \label{lem:closure-recognize}
  $\C{\N\Sigma} \subseteq \N{\C\Sigma}$ for every language $\Sigma \subseteq S^*$.
\end{lemma}
\begin{proof}
  Every configuration in $\C{\N\Sigma}$ comes from
  $$
    \N\Sigma \ni \sigma \move[\eta\in\Eps] \sigma'
    \quad\implies\quad
    \sigma' \in \C{\N\Sigma}.
  $$
  Let $\eta = \eta\red\con\eta\nred$ be the canonical decomposition of $\eta$, witnessed by $\sigma\red$, $s\nred$ and $\sigma\nred$.
  If $\eta\nred = \epsilon$, then $\sigma = \sigma\red \con \sigma$
  with $\sigma\red \in \N{\sigma_1}$ and $\sigma' \in \N{\sigma_2}$ for some $\sigma_1 \con \sigma_2 \in \Sigma$.
  Then $\sigma_1 \move[\Eps] \epsilon$, implying that $\sigma_2 \in \C\Sigma$ and $\sigma' \in \N{\C\Sigma}$.

  Otherwise, for some $\sigma''$ we must have
  $$
    \sigma = \sigma\red \con s\nred \con \sigma'' \move[\eta\red] s\nred \con \sigma'' \move[\eta\nred] \sigma\nred \con \sigma'' = \sigma'.
  $$
  Also, there must exist $\sigma_1 \con s\nred \con \sigma_2 \in \Sigma$
  with $\sigma\red \in \N{\sigma_1}$ and $\sigma'' \in \N{\sigma_2}$.
  We then have
  $$
    \sigma_1 \con s\nred \con \sigma_2 \move[\Eps] s\nred \con \sigma_2 \move[\Eps] \sigma\nred \con \sigma_2,
  $$
  from which $\sigma\nred \con \sigma_2 \in \C\Sigma$
  and $\sigma' = \sigma\nred \con \sigma'' \in \N{\C\Sigma}$ as desired.
\end{proof}

\begin{proposition}
  \label{prop:main-recognize}
  $\N{\Sigma_{\epsilon}} = \N{\C{\sstart\con\sstop}}$ and $\N{\Sigma_{\alpha a}} = \N{\C{\S[a]{\N{\Sigma_{\alpha}}}}}$.
\end{proposition}
\begin{proof}
  The first equality follows directly from the definition of $\Sigma_{\epsilon}$.
  The second can be obtained by combining Lemmas~\ref{lem:shift-recognize} and~\ref{lem:closure-recognize},
  idempotency of null closure, and monotonicity of null closure, shift and call-reduce closure into the following:
  $$
    \N{\C{\S[a]{\N{\Sigma_{\alpha}}}}} \subseteq
    \N{\C{\N{\S[a]{\Sigma_{\alpha}}}}} \subseteq
    \N{\N{\C{\S[a]{\Sigma_{\alpha}}}}} =
    \N{\N{\Sigma_{\alpha a}}} =
    \N{\Sigma_{\alpha a}} =
    \N{\C{\S[a]{\Sigma_{\alpha}}}} \subseteq
    \N{\C{\S[a]{\N{\Sigma_{\alpha}}}}}.
  $$
\end{proof}

As a consequence, plugging Algorithm~\ref{alg:phase-recognize-deep} into Algorithm~\ref{alg:main-recognize} (with all call-reduce closures replaced by their null closures)
we obtain a procedure for context-free recognition,
realizable using any data type $\type\Sigma$ representing languages over $S$, implementing API~\ref{api:languages}.

\begin{api}\label{api:languages}
  Operations required for recognition.
  \begin{itemize}
    \item $\texttt{root}() \in \type\Sigma$ \tab $\{\epsilon\}$,
    \item $\texttt{prepend}(\N{\C{s}}; \Sigma \in \type\Sigma) \in \type\Sigma$ \tab $\N{\C{s}} \con \Sigma$,
    \item $\texttt{union}(\Sigma_1,\Sigma_2 \in \type\Sigma) \in \type\Sigma$ \tab $\Sigma_1 \cup \Sigma_2$,
    \item $\texttt{epsilon}(\Sigma \in \type\Sigma) \in \{\texttt{true},\texttt{false}\}$ \tab whether $\epsilon \in \Sigma$,
    \item $\texttt{derivative}(\Sigma \in \type\Sigma; s \in S) \in \type\Sigma^{\leqslant k}$ \tab $\Sigma\deriv{s}$.
  \end{itemize}
\end{api}

\section{Parsing}
\label{sec:parsing}

We now augment our algorithm to perform generalized parsing and not just mere recognition.
The exposition here mirrors closely the one from the previous section.
All definitions and algorithms are adjusted so that instead of asking \emph{whether} some configuration is reachable,
we ask about \emph{the language of transitions} by which it can be reached.
In particular, instead of computing (a representation of) all possible parse trees,
we are going to output the language of all legal transition sequences.

\begin{definition}
  For an input $\alpha \in A^*$,
  the \bld{parsing relation} $\Phi_{\alpha} \subseteq S^* {\times} D^*$
  associates configurations reachable when parsing $\alpha$
  with \textit{reversed}%
  \footnote{Tracking how an RTN configuration evolves under a sequence of transitions,
  one can see that the states at the \emph{beginning} (top) of the ultimate configuration depend primarily on the transitions at the \emph{end} of the sequence.
  Thus, to keep the direction of concatenation consistent, we record the transition history in reverse.}
  sequences of transitions by which they can be reached from the initial configuration:
  $$
    \Phi_{a_1 \cons a_n} \deq \{(\sigma\con\sstop,\rev\delta):\ \sstart \move[\delta \in \Eps \con D_{a_1} \cons \Eps \con D_{a_n} \con \Eps ] \sigma \}.
  $$
\end{definition}

In particular, accepting transition sequences for an input string $\alpha$
are exactly (reversed) those which $\Phi_{\alpha}$ associates with the configuration $\sstop$.

\begin{definition}
  The \bld{call-reduce closure} and \bld{shift by a symbol $a \in A$}
  of a relation $\Phi \subseteq S^* {\times} D^*$
  are defined as:
  \begin{align*}
    \C\Phi &\deq \{(\sigma', \rev\eta \con \rev\delta):\ \sigma \move[\eta \in \Eps] \sigma',\ (\sigma, \rev\delta)\in\Phi  \}, \\
    \S[a]\Phi &\deq \{(\sigma', d \con \rev\delta):\ \sigma \move[d \in D_a] \sigma',\ (\sigma, \rev\delta)\in\Phi \}.
  \end{align*}
\end{definition}

Parsing relations for a particular input can be calculated inductively analogously to parsing languages,
because $\Phi_{\epsilon} = \C{\sstart\con\sstop}$ and $\Phi_{\alpha a} = \C{\S[a]{\Phi_{\alpha}}}$.
To efficiently perform a single phase according to the latter equality, we need a suitable notion of derivatives:

\begin{definition}
  \label{df:relderiv}
  A \bld{derivative of degree $k$} of a relation $\Phi \subseteq S^* {\times} D^*$ by a symbol $s \in S$
  is any collection of pairs $\Phi\deriv{s} =\{(\Delta_i,\Phi_i)\}_{i = 1 \ldots k}$
  with $\Delta_i \subseteq D^*$ and $\Phi_i \subseteq S^* {\times} D^*$,
  such that for each $\sigma \in S^*$ and $\delta \in D^*$ we have
  $$
    (s\con\sigma, \rev\delta) \in \Phi \iff (s\con\sigma, \rev\delta) \in \bigcup_i \Delta_i \con (s,\epsilon) \con \Phi_i.
  $$
  A relation is \bld{$(m,k)$-regular} if it belongs to some family of size $m$, closed under derivatives of degree $k$.
  It is \bld{$k$-regular} if it is $(m,k)$-regular for some $m$, and \bld{regular} if it is $k$-regular for some $k$.
\end{definition}

Regular relations, as defined here, are similar to \emph{rational relations}---%
a well-studied generalization of regular languages~\cite{rational}.
The fundamental difference is that the latter require the components to be "jointly regular",
whereas we put no restriction on the transition languages $\Delta_i$.
We are only concerned about finiteness (and bounded nondeterminism) of the automaton with respect to $S$.
Note however, that despite this flexibility, such automata are not always determinizable.

\begin{algorithm}\label{alg:phase-parse-deriv}
  Calculating $\C{\S[a]\Phi}$
  \begin{lstlisting}
  $\Phi\out \assign \varnothing$
  foreach $d = \shift{s,a,t} \in D_a$:
    foreach $(\Delta_0,\Phi_0) \in \Phi\deriv{s}$:
      $\Phi\out \assign \Phi\out \ \cup\  \C{t} \con d \con \Delta_0 \con \Phi_0$
      if $\nulling{t} \neq \varnothing$, foreach $s_1 \in S$:
        foreach $(\Delta_1, \Phi_1) \in \Phi_0\deriv{s_1}$:
          $\Phi\out \assign \Phi\out \ \cup\  \C{s_1} \con \nulling{t} \con d \con \Delta_0 \con \Delta_1 \con \Phi_1$
          if $\nulling{s_1} \neq \varnothing$, foreach $s_2 \in S$:
            foreach $(\Delta_2, \Phi_2) \in \Phi_1\deriv{s_2}$:
              $\Phi\out \assign \Phi\out \ \cup\  \C{s_2} \con \nulling{s_1} \con \nulling{t} \con d \con \Delta_0 \con \Delta_1 \con \Delta_2 \con \Phi_2$
              ...
  return $\Phi\out$
  \end{lstlisting}
\end{algorithm}

For convenience, let us denote
$$
  \D{s}\sigma \deq \{\rev\eta :\  s \move[\eta\in\Eps] \sigma \}.
$$
Algorithm~\ref{alg:phase-parse-deriv} (completely analogous to Algorithm~\ref{alg:phase-recognize-deriv}) can now be used to calculate $\C{\S[a]\Phi}$.
Adapting it to have constant depth is however not as easy as before,
because skipping nullable states requires additional transitions to be performed.
To handle this, we enrich the alphabet of configurations to
$$
  G \deq S \cup D_\epsilon,
$$
allowing them to intersperse states with (non-shift) transitions.
Such transitions will "guide" the simulation of the RTN---%
if present on top of a configuration, they will need to be executed (en bloc) before the states below them can be accessed.
To express that, reachability is enriched by
$$
  \eta \con \gamma \move[\eta \in D_\epsilon^+] \gamma \quad\text{if } \gamma \text{ does not start with a transition.}
$$
Note, that guiding transitions will be introduced into configurations only artificially, replacing states which we commit to null later on---%
reachability may only remove them, fulfilling that commitment.

We now consider \bld{guided languages} (languages of guided configurations, i.e., subsets of $G^*$)
and \bld{guided relations} (subsets of $G^* {\times} D^*$),
and (re)define null closure to optionally replace states by their nulling transitions:

\begin{definition}
  A guided language $\Gamma \subseteq G^*$  (respectively, guided relation $\Psi \subseteq G^* {\times} D^*$) is \bld{null-closed} if
  for each $S \ni s \move[\eta\in\Eps] \epsilon$, $\gamma_1,\gamma_2 \in G^*$ (and $\delta \in D^*$) we have
  \begin{align*}
    \gamma_1 \con s \con \gamma_2 \in \Gamma &\implies \gamma_1 \con \eta \con \gamma_2 \in \Gamma\\
    \text{(respectively,}\quad (\gamma_1 \con s \con \gamma_2, \rev\delta) \in \Psi &\implies  (\gamma_1 \con \eta \con \gamma_2, \rev\delta) \in \Psi \quad\text{).}
  \end{align*}
  The \bld{null closure} of a language $\Gamma$ (relation $\Psi$) is the smallest null-closed guided language $\N\Gamma \supseteq \Gamma$ (guided relation $\N\Psi \supseteq \Psi$).
\end{definition}

We now revisit and augment the reasoning from the previous section
to show that an inductive procedure can be used to compute the null closures of parsing relations.
First, we make a convenient observation about the interplay between performing transitions and null closure:

\begin{lemma}
  \label{lem:track-back}
  Consider $\gamma \in G^*$ and $\gamma' \in \N\gamma$.
  Then any sequence of transitions $\gamma' \move[\delta \in D^*] \gamma'_2$
  can be "tracked back" to $\gamma \move[\delta] \gamma_2$ with $\gamma'_2 \in \N{\gamma_2}$.
\end{lemma}
\begin{proof}
  The proof is by induction on the length of $\delta$.
  The case when $\delta = \epsilon$ is trivial.
  Otherwise, $\gamma'$ and $\gamma$ must be nonempty.
  We distinguish three cases.

  If $\gamma = d\con\gamma_1$ for some $d \in D_\epsilon$ and $\gamma_1 \in G^*$, we have
  $$
    \gamma' = d\con\gamma'_1 \move[d] \gamma'_1 \move[\delta_1] \gamma'_2
    \quad\text{and}\quad
    \gamma = d\con\gamma_1 \move[d] \gamma_1,
  $$
  and the result follows from applying the induction hypothesis to $\N{\gamma_1} \ni \gamma'_1 \move[\delta_1] \gamma'_2$.

  If $\gamma = s\con\gamma_1$ for some $s \in S$ and $\gamma'$ also begins with a state, it must be the same state $s$.
  Then the first transition $d$ in $\delta$ must be valid from $s$, giving
  $$
    \gamma' = s\con\gamma'_1 \move[d] \sigma\con\gamma'_1 \move[\delta_1] \gamma'_2
    \quad\text{and}\quad
    \gamma = s\con\gamma_1 \move[d] \sigma\con\gamma_1,
  $$
  so we can succeed by applying the induction hypothesis to $\N{\sigma\con\gamma_1} \ni \sigma\con\gamma'_1 \move[\delta_1] \gamma'_2$.

  Finally, we might have $\gamma = s\con\gamma_1$ and $\gamma' = \eta\con\gamma'_1$ with $\gamma'_1 \in \N{\gamma_1}$ and $s \move[\eta] \epsilon$.
  Then $\delta = \eta\con\delta_1$ (it cannot be shorter than $\eta$, because guiding transitions can only be executed en bloc),
  resulting in
  $$
    \gamma' = \eta\con\gamma'_1 \move[\eta] \gamma'_1 \move[\delta_1] \gamma'_2
    \quad\text{and}\quad
    \gamma = s\con\gamma_1 \move[\eta] \gamma_1.
  $$
  We finish by invoking the inductive hypothesis on $\N{\gamma_1} \ni \gamma'_1 \move[\delta_1] \gamma'_2$.
\end{proof}

\begin{corollary}
  \label{lem:shift-parsing}
  $\S[a]{\N\Psi} \subseteq \N{\S[a]\Psi}$ for every guided relation $\Psi \subseteq G^* {\times} D^*$.
\end{corollary}

\begin{corollary}
  \label{lem:closure-parsing}
  $\C{\N\Psi} \subseteq \N{\C\Psi}$ for every guided relation $\Psi \subseteq G^* {\times} D^*$.
\end{corollary}

\begin{corollary}
  Null closure preserves being call-reduce-closed:
  $$
    \C\Psi = \Psi \implies \C{\N\Psi} = \N\Psi.
  $$
\end{corollary}
\begin{proof}
  $\C{\N\Psi} = \C{\N{\C\Psi}} \subseteq \N{\C{\C\Psi}} = \N{\C\Psi} = \N\Psi \subseteq \C{\N\Psi}$.
\end{proof}

\begin{proposition}
  \label{prop:main-parse}
  $\N{\Phi_{\epsilon}} = \N{\C{\sstart\con\sstop}}$ and $\N{\Phi_{\alpha a}} = \N{\C{\S[a]{\N{\Phi_{\alpha}}}}}$.
\end{proposition}
\begin{proof}
  Redo the proof of Proposition~\ref{prop:main-recognize},
  using Corollaries~\ref{lem:shift-parsing} and~\ref{lem:closure-parsing} instead of Lemmas~\ref{lem:shift-recognize} and~\ref{lem:closure-recognize}.
\end{proof}

To proceed to phase computation, we define two variants of derivatives for guided relations:
the first one, analogous to Definition~\ref{df:relderiv}, separates the guiding transitions,
while the second, weaker, concatenates them to the remaining relation:

\begin{definition}
  A \bld{derivative of degree $k$} of a guided relation $\Psi \subseteq G^* {\times} D^*$ by a symbol $s \in S$
  is any collection of pairs $\Psi\deriv{s} = \{(\Lambda_i,\Psi_i)\}_{i = 1 \ldots k}$
  with $ \Lambda_i \subseteq D^* {\times} D^*$ and $\Psi_i \subseteq G^* {\times} D^*$,
  such that for each $\eta \in \Eps,\ \gamma \in G^*$ and $\delta \in D^*$ we have
  $$
    (\eta \con s \con \gamma, \rev\delta) \in \Psi  \iff  (\eta \con s \con \gamma, \rev\delta) \in \bigcup_i \Lambda_i \con (s,\epsilon) \con \Psi_i.
  $$
  A guided relation is \bld{$(m,k)$-regular} if it belongs to some family of size $m$, closed under derivatives of degree $k$.
  It is \bld{$k$-regular} if it is $(m,k)$-regular for some $m$, and \bld{regular} if it is $k$-regular for some $k$.
\end{definition}

\begin{definition}
  A \bld{flat derivative of degree $k$} of a guided relation $\Psi \subseteq G^* {\times} D^*$ by a symbol $s \in S$
  is any collection of guided relations $\Psi\deriv[flat]{s} = \{\Psi_i\}_{i = 1 \ldots k}$
  with $\Psi_i \subseteq G^* {\times} D^*$,
  such that for each $\gamma \in G^*, \delta' \in D^*$:
  $$
    (s \con \gamma, \rev{\delta'}) \in \bigcup_i (s,\epsilon) \con \Psi_i
    \iff
    \text{ for some } \rev\eta \con \rev\delta = \rev{\delta'} \text{ we have }
    (\eta \con s \con \gamma, \rev\delta) \in \Psi.
  $$
\end{definition}

\begin{algorithm}\label{alg:phase-parse-deep}
  Calculating $\N{\C{\S[a]\Psi}}$ for call-reduce-closed, proper $\Psi$
  \begin{lstlisting}
  $\Psi\out \assign \varnothing$
  foreach $d = \shift{s,a,t} \in D_a$:
    foreach $\Psi_0 \in \Psi\deriv[flat]{s}$:
      $\Psi\out \assign \Psi\out \ \cup\  \N{\C{t}} \con d \con \Psi_0$
      if $\nulling{t} \neq \varnothing$, foreach $s' \in S$:
        foreach $\Psi_1 \in (\nulling{t} \con d \con \Psi_0)\deriv[flat]{s'}$:
          $\Psi\out \assign \Psi\out \ \cup\  \N{\C{s'}} \con \Psi_1$
  return $\Psi\out$
  \end{lstlisting}
\end{algorithm}

When transitions are forgotten, both above definitions degenerate to Brzozowski derivatives.
The distinction is necessary for clean separation of concerns in our algorithms,
and for their complexity analysis.
In particular, only the weaker notion of derivatives is needed to calculate null closures of parsing relations,
as shown in Algorithm~\ref{alg:phase-parse-deep}.
We only prove its correctness for guided relations which are \bld{proper},
i.e., can be obtained as null closures of non-guided relations.

\begin{proposition}
  \label{lem:phase-parse}
  For a call-reduce-closed, proper $\Psi = \N\Phi$
  in which every configuration has the state $\sstop$ at its bottom,
  Algorithm~\ref{alg:phase-parse-deep} correctly returns $\Psi\out = \N{\C{\S[a]\Psi}}$.
\end{proposition}
\begin{proof}
  \bld{Claim 1: $\C{\S[a]\Phi} \subseteq \Psi\out$.}
  Every pair in $\C{\S[a]\Phi}$ must come from
  $$
    (s\con\sigma, \rev\delta) \in \Phi \subseteq \Psi, \quad
    s\con\sigma \move[d = \shift{s,a,t}] t\con\sigma \move[\eta] \sigma'
    \quad\implies\quad
    (\sigma', \rev\eta \con d \con \rev\delta) \in \C{\S[a]\Phi}.
  $$

  Let $\eta = \eta\red\con\eta\nred$ be the canonical decomposition of $\eta$, witnessed by $\sigma\red$, $s\nred$ and $\sigma\nred$.
  If $\eta\red = \epsilon$, we have $s\nred = t$ and $\sigma' = \sigma\nred\con\sigma$,
  from where
  \begin{alignat*}{3}
    (\sigma, &&\rev\delta) &\ \in \Psi_0, &\text{(for some } \Psi_0 \in \Psi\deriv[flat]{s} \text{)} \\
    (\sigma\nred, &&\rev\eta\nred) &\ \in \C{t} \subseteq \N{\C{t}}, \\
    (\sigma',&&\ \rev\eta \con d \con \rev\delta) &\ \in \N{\C{t}} \con d \con \Psi_0 \subseteq \Psi\out.
  \end{alignat*}.

  Otherwise, for some $\sigma''$, we have $t\con\sigma = \sigma\red \con s\nred \con \sigma''$ and $\sigma' = \sigma\nred\con\sigma''$
  (the state $s\nred$ surely exists, as every configuration has non-nullable $\sstop$ at its bottom).
  Because $\sigma\red$ begins with $t$, we might split $\sigma\red \move[\eta\red] \epsilon$ as $\sigma\red = t \con \sigma_* \move[\eta_t] \sigma_* \move[\eta_*] \epsilon$,
    obtaining
  \begin{alignat*}{3}
    (s \con \sigma_* \con s\nred \con \sigma'', &&\rev\delta) &\ \in \Phi, \\
    (s \con \eta_* \con s\nred \con \sigma'', &&\rev\delta) &\ \in \N\Phi = \Psi, \\
    (\eta_* \con s\nred \con \sigma'', &&\rev\delta) &\ \in \Psi_0, &\text{(for some } \Psi_0 \in \Psi\deriv[flat]{s} \text{)} \\
    (\eta_* \con s\nred \con \sigma'', &&\rev{\eta_t} \con d \con \rev\delta) &\ \in \nulling{t} \con d \con \Psi_0, \\
    (\sigma'', &&\ \rev{\eta_*} \con \rev{\eta_t} \con d \con \rev\delta) &\ \in \Psi_1, &\text{(for some } \Psi_1 \in (\nulling{t} \con d \con \Psi_0)\deriv[flat]{s\nred} \text{)} \\
    (\sigma\nred, &&\rev{\eta\nred}) &\ \in \C{s\nred} \subseteq \N{\C{s\nred}}, \\
    (\sigma', &&\rev\eta \con d \con \rev\delta) &\ \in \N{\C{s\nred}} \con \Psi_1 \subseteq \Psi\out.
  \end{alignat*}

  \bld{Claim 2: $\Psi\out$ is null-closed.}
  This is because it is obtained from null-closed pieces
  ($\Psi = \N\Phi$ and $\N{\C{s}}$ are by definition, while $\nulling{t}$ and $d$ consist only of transitions)
  by operations which preserve being null-closed (concatenation, union, and derivatives).

  \bld{Claim 3: $\Psi\out \subseteq \N{\C{\S[a]\Psi}}$.}
  Consider any pair added to $U$.
  If it was added in line~4, it must be of the form
  $$
    (\gamma_1 \con \gamma_2,\ \rev{\eta_1} \con d \con \rev{\delta_2})
  $$
  for $(\gamma_1, \rev{\eta_1}) \in \N{\C{t}}$, $d = \shift{s,a,t}$ and $(\gamma_2, \rev{\delta_2}) \in \Psi_0$.
  We now have:
  \begin{alignat*}{3}
    (\eta \con s \con \gamma_2, && \rev\delta) &\ \in \Psi, &(\text{for some } \rev\eta \con \rev\delta = \rev{\delta_2})\\
    (s \con \gamma_2, && \rev\eta \con \rev\delta) &\ \in \Psi, \\
    (t \con \gamma_2, && d \con \rev\eta \con \rev\delta) &\ \in \S[a]\Psi, \\
    (\gamma_1 \con \gamma_2, &&\ \rev{\eta_1} \con d \con \rev{\delta_2}) &\ \in \N{\C{\S[a]\Psi}}.
  \end{alignat*}

  Pairs added in line~7 must be of the form
  $$
    (\gamma_1 \con \gamma_2, \rev{\eta_1} \con \rev{\delta_2})
  $$
  for $(\gamma_1, \rev{\eta_1}) \in \N{\C{s'}}$ and $(\gamma_2, \rev{\delta_2}) \in \Psi_1$.
  In this case we have:
  \begin{alignat*}{3}
    (\eta_3 \con s' \con \gamma_2, && \rev{\delta_3}) &\ \in \nulling{t} \con d \con \Psi_0, &(\text{for some } \rev{\eta_3} \con \rev{\delta_3} = \rev{\delta_2})\\
    (\eta_3 \con s' \con \gamma_2, && \rev{\delta_4}) &\ \in \Psi_0, &(\text{for some } \rev\eta \in \nulling{t},\ \rev\eta \con d \con \rev{\delta_4} = \rev{\delta_3})\\
    (\eta_5 \con s \con \eta_3 \con s' \con \gamma_2, && \rev{\delta_5}) &\ \in \Psi, &(\text{for some } \rev{\eta_5} \con \rev{\delta_5} = \rev{\delta_4})\\
    (s \con \eta_3 \con s' \con \gamma_2, && \rev{\eta_5} \con \rev{\delta_5}) &\ \in \Psi, \\
    (t \con \eta_3 \con s' \con \gamma_2, && d \con \rev{\delta_4}) &\ \in \S[a]\Psi, \\
    (s' \con \gamma_2, &&\ \rev{\eta_3} \con \rev\eta \con d \con \rev{\delta_4}) &\ \in \C{\S[a]\Psi}, \\
    (\gamma_1 \con \gamma_2, && \rev{\eta_1} \con \rev{\delta_2}) &\ \in \N{\C{\S[a]\Psi}}.
  \end{alignat*}

  \bld{Finale:}
  Combine the above claims with Corollaries~\ref{lem:shift-parsing} and~\ref{lem:closure-parsing} to obtain
  $$
    \N{\C{\S[a]\Psi}} = \N{\C{\S[a]{\N\Phi}}} \subseteq \N{\N{\C{\S[a]\Phi}}} \subseteq \N{\N\Psi\out} = \Psi\out \subseteq \N{\C{\S[a]\Psi}}.
  $$
\end{proof}

\begin{algorithm}\label{alg:phase-parse-final}
  Calculating $\N{\C{\S[a]\Psi}}$ for call-reduce-closed, proper $\Psi$ (final)
  \begin{lstlisting}
  $\Psi\out, \Psi\out[aux] \assign \varnothing$
  foreach $d = \shift{s,a,t} \in D_a$:
    foreach $\Psi_0 \in \Psi\deriv[flat]{s}$:
      $\Psi\out \assign \Psi\out \ \cup\  \N{\C{t}} \con d \con \Psi_0$
      if $\nulling{t} \neq \varnothing$:
        $\Psi\out[aux] \assign \Psi\out[aux] \ \cup\ \nulling{t} \con d \con \Psi_0$
  foreach $s' \in S$:
    foreach $\Psi_1 \in \Psi\out[aux]\deriv[flat]{s'}$:
      $\Psi\out \assign \Psi\out \ \cup\  \N{\C{s'}} \con \Psi_1$
  return $\Psi\out$
  \end{lstlisting}
\end{algorithm}

One last observation we apply is that derivatives distribute over union,
and thus the inner loops in lines~5~and~6 can be pulled out to the top level,
resulting in Algorithm~\ref{alg:phase-parse-final}.

\begin{algorithm}\label{alg:main-parse}
  Parsing $\alpha \in A^*$
  \begin{lstlisting}
  $\Psi \assign \N{\C\sstart} \con \N{\C\sstop}$
  foreach $a \in \alpha$ (in order):
    $\Psi \assign \N{\C{\S[a]\Psi}}$
  $\Delta\out \assign \varnothing$
  foreach $\Psi' \in \Psi\deriv[flat]{\sstop}$:
    $\Delta\out \assign \Delta\out \ \cup\  \{ \delta\con\eta :\  (\eta, \rev\delta) \in \Psi', \eta \in \Eps \}$
  return $\Delta\out$
  \end{lstlisting}
\end{algorithm}

Extracting the language of accepting transition sequences from the parsing relation for the whole input is also a bit more complicated,
because any remaining guiding transitions need to be executed.
Thus we replace our main Algorithm~\ref{alg:main-recognize} with Algorithm~\ref{alg:main-parse}.
Overall, we realize generalized context-free parsing,
using any representation $\type\Psi$ of guided relations supporting API~\ref{api:relations}.

\begin{api}\label{api:relations}
  Operations on relations required for full parsing.
  \begin{itemize}
    \item $\texttt{root}() \in \type\Psi$ \tab $\epsilon$,
    \item $\texttt{prepend}(\N{\C{s}}; \Psi \in \type\Psi) \in \type\Psi$ \tab $\N{\C{s}} \con \Psi$,
    \item $\texttt{prepend}(\Delta \subseteq D^*; \Psi \in \type\Psi) \in \type\Psi$ \tab $\Delta \con \Psi$,
    \item $\texttt{union}(\Psi_1, \Psi_2 \in \type\Psi) \in \type\Psi$ \tab $\Psi_1 \cup \Psi_2$,
    \item $\texttt{epsilon}(\Psi \in \type\Psi) \subseteq D^*$ \tab $\{ \delta\con\eta: (\eta,\rev\delta)\in\Psi \}$,
    \item $\texttt{derivative}(\Psi \in \type\Psi; s \in S) \in \type\Psi^{\leqslant k}$ \tab $\Psi\deriv[flat]{s}$.
  \end{itemize}
\end{api}

\section{Representations}
\label{sec:representations}

We now make concrete the key observation underlying this paper:
that although the sets of reachable parser configurations are sometimes infinite,
their regularity makes it possible to manipulate them efficiently.

\subsection{Parsing languages and relations}
\label{sec:dag}

\begin{api}\label{api:atomic-languages}
  Required operations on atomic languages.
  \begin{itemize}
    \item $\texttt{atomic}(s \in S) \in \type\Upsilon$ \tab $\N{\C{s}}$,
    \item $\texttt{epsilon}(\Upsilon \in \type\Upsilon) \in \{\texttt{true},\texttt{false}\}$ \tab whether $\epsilon \in \Upsilon$,
    \item $\texttt{derivative}(\Upsilon \in \type\Upsilon,  s \in S) \in \type\Upsilon^{\leqslant k}$ \tab $\Upsilon\deriv{s}$.
  \end{itemize}
\end{api}

We begin top-down, by presenting a data structure implementing API~\ref{api:languages}.
To make it independent of the grammar, let us parameterize it with a data type $\type\Upsilon$,
representing a finite family of languages over $S$, realizing API~\ref{api:atomic-languages}.
We allow the first argument of \texttt{prepend} to be any element of $\type\Upsilon$.

Denote by $\efree\Sigma$ the $\epsilon$-free part of $\Sigma$ (i.e., $\Sigma - \{\epsilon\}$).
Our implementation is based on the following observation:

\begin{proposition}
  \label{prop:language-shape}
  Let $\Sigma_1, \Sigma_2, \ldots, \Sigma_n$ be a sequence of languages over $S$,
  in which each $\Sigma_i$ is obtained by applying one of the operations of API~\ref{api:languages} to some $\Sigma_j$ with $j < i$.
  Then there exists a finite collection of pairs $\{ (\Upsilon, i_j) \}_j$ with $\Upsilon_j \in \type\Upsilon$ and $1 \leqslant i_j < n$, such that
  \[
    \efree{\Sigma_n} = \bigcup_j \efree{\Upsilon_j} \con \Sigma_{i_j}.
  \]
\end{proposition}
\begin{proof}
  We argue by induction on $n$.
  For languages returned by \texttt{derivative}, the inductive step follows by
  \[
    \efree{\Sigma_n\deriv{s}} \ =\
    \efree{\efree{\Sigma_n}\deriv{s}} \ =\
    \efree{\bigcup_j \bigcup_{\Upsilon' \in \efree{\Upsilon_j}\deriv{s}} \Upsilon' \con \Sigma_{i_j}} \ =\
    \bigcup_{\substack{j \\ \Upsilon' \in \Upsilon_j\deriv{s}}} \efree{\Upsilon'} \con \Sigma_{i_j}
      \ \cup\  \bigcup_{\substack{j \\ \epsilon \in \Upsilon' \in \Upsilon_j\deriv{s}}} \efree{\Sigma_{i_j}},
  \]
  and applying the inductive hypothesis to the latter part of the sum. All other cases are trivial.
\end{proof}

Our proposed data structure keeps, for each $\Sigma_n$, the set of pairs warranted above,
and an additional bit denoting whether $\epsilon \in \Sigma_n$.
All operations have natural implementations, with \texttt{derivative} following the above proof.
Their complexity will be analyzed in Section~\ref{sec:complexity}.
Here, we note only that, as language union is idempotent, any duplicate pairs encountered can and should be merged.
Thus, each $\Sigma_n$ is represented by at most $n\,|\type\Upsilon|$ pairs.

The whole structure can be viewed as a DAG, whose vertices correspond to the languages $\Sigma_i$.
Edges, labeled with elements of $\type\Upsilon$, follow the references from Proposition~\ref{prop:language-shape}:
an edge $(u, v)$ with label $\Upsilon$ is present iff the decomposition of $\efree{\Sigma_u}$ contains the pair $(\Upsilon, v)$.
Each vertex is additionally labeled with the $\epsilon$ bit---vertices at which this bit is set will be called \bld{final}.

\begin{api}\label{api:atomic-relations}
  Required operations on atomic relations
  ($\type\Lambda^\circ$ represents relations over $D^*$, i.e., subsets of $D^*{\times}D^*$).
  \begin{itemize}
    \item $\texttt{atomic}(s \in S) \in \type\Xi$ \tab $\N{\C{s}}$,
    \item $\texttt{epsilon}(\Xi \in \type\Xi) \in \type\Lambda^\circ$ \tab $\Xi \cap D^*{\times}D^*$,
    \item $\texttt{derivative}(\Xi \in \type\Xi, s \in S) \in (\type\Lambda^\circ \times \type\Xi)^{\leqslant k}$ \tab $\Xi\deriv{s}$.
  \end{itemize}
\end{api}

In order to adapt this data structure to represent parsing relations, we need to additionally keep track of transitions.
Instead of the underlying type $\type\Upsilon$, let us parameterize it with a type $\type\Xi$,
representing a finite family of guided relations, realizing API~\ref{api:atomic-relations}.
Our implementation is supposed to compute flat derivatives,
let us then introduce an auxiliary \bld{flattening} operation, defined for any $\Lambda \subseteq D^*{\times}D^*$ as
\[
  \flat\Lambda \deq \{ \rev{\delta \con \eta} :\ (\eta, \rev\delta) \in \Lambda \}.
\]

With $\efree\Psi$ denoting the $\epsilon$-free part of $\Psi$ (i.e., $\Psi - D^*{\times}D^*$),
we can now show the following generalization of Proposition~\ref{prop:language-shape}:

\begin{proposition}
  \label{prop:relation-shape}
  Let $\Psi_1, \Psi_2, \ldots, \Psi_n$ be a sequence of guided relations over $G^*{\times}D^*$,
  in which each $\Psi_i$ is obtained by applying one of the operations of API~\ref{api:relations} to some $\Psi_j$ with $j < i$.
  Then there exists a finite collection of triples $\{ (\Lambda_j, \Xi_j , i_j) \}_j$
  with $\Lambda_j \subseteq D^*{\times}D^*$, $\Xi_j \in \type\Xi$ and $1 \leqslant i_j < n$, such that
  \[
    \efree{\Psi_n} = \bigcup_j \Lambda_j \con \efree{\Xi_j} \con \Psi_{i_j}.
  \]
\end{proposition}
\begin{proof}
  Again, the only nontrivial operation to be considered in the inductive proof is \texttt{derivative}.
  This time the inductive step follows from
  \begin{align*}
    \efree{\Psi_n\deriv[flat]{s}} \ &=\
    \efree{\bigcup_{\substack{j \\ (\Lambda,\Xi') \in \efree{\Xi_j}\deriv{s}}} \flat{\Lambda_j \con \Lambda} \con \Xi' \con \Psi_{i_j}} \ = \\\ \\
    &=\ \bigcup_{\substack{j \\ (\Lambda,\Xi') \in \Xi_j\deriv{s}}} \flat{\Lambda_j \con \Lambda} \con \efree{\Xi'} \con \Psi_{i_j}
         \ \cup\ \bigcup_{\substack{j \\ (\Lambda,\Xi') \in \Xi_j\deriv{s}}} \flat{\Lambda_j \con \Lambda} \con (\Xi' \cap D^*{\times}D^*) \con \efree{\Psi_{i_j}}.
  \end{align*}
\end{proof}

Applying this result to our DAG structure requires labeling edges with pairs $(\Lambda, \Xi)$, with $\Lambda \subseteq D^*{\times}D^*$ and $\Xi \in \type\Xi$.
Instead of the single $\epsilon$ bit, vertices need to be labeled with relations over $D^*$ (subsets of $D^*{\times}D^*$)---%
\bld{final} vertices are now those with a nonempty label.

As we can see, for edge and vertex labels we need an additional data type, denoted $\type\Lambda$, representing relations over $D^*$.
For efficiency reasons, we allow it to be different from the type $\type\Lambda^\circ$ returned from \texttt{derivative} and \texttt{epsilon} calls on atomic relations.
Realizing our API requires elements of $\type\Lambda$ to be merged using set union, flattened.
We also require concatenatenation on the left with an atomic $\type\Lambda^\circ$ or with whatever $\Delta \subseteq D^*$ our \texttt{prepend} method is called.
To simplify further reasoning, we combine them into a single operation.
Concatenation on the right with elements of $\type\Lambda^\circ$ is always immediately followed by flattening,
and is therefore redundant, as $\flat{\Lambda \con \Lambda^\circ} = \flat{\flat\Lambda \con \Lambda^\circ \con \epsilon}$.
The requirements are summarized by API~\ref{api:transition-relations}.

\begin{api}\label{api:transition-relations}
  Required operations on transition relations (subsets of $D^*{\times}D^*$).
  $\type\Delta$ represents languages of transitions (i.e., subsets of $D^*$).
  \begin{itemize}
    \item $\texttt{epsilon}() \in \type\Lambda$ \tab $\epsilon$,
    \item $\texttt{prepend}(\Delta \in \type\Delta, \Lambda^\circ \in \type\Lambda^\circ, \Lambda \in \type\Lambda) \in \type\Lambda$ \tab $\Delta \con \Lambda^\circ \con \Lambda$,
    \item $\texttt{flatten}(\Lambda \in \type\Lambda) \in \type\Delta$ \tab $\flat{\Lambda}$,
    \item $\texttt{union}(\Lambda_1, \ldots, \Lambda_k \in \type\Lambda) \in \type\Lambda$ \tab $\Lambda_1 \cup \cdots \cup \Lambda_k$.
  \end{itemize}
\end{api}

It is also clear that if the Boolean type $\type{B}$ is used in place of $\type\Lambda$ (and thus also $\type\Lambda^\circ$ and $\type\Delta$),
we are left with the simpler DAG useful for recognition.

\subsection{Atomic closures}

We now turn to the call-reduce closures of individual states, realizing APIs~\ref{api:atomic-languages} and~\ref{api:atomic-relations}.
The constructive proof of the proposition below presents a possible direct representation using nondeterministic finite automata.
The non-relational variant could of course be determinized and minimized,
but it might not be advantageous to do so---%
the price for less nondeterminism is a larger automaton size,
and the grammar-related complexity of our complete algorithm depends on both.

Let us denote
$$
  \P{s}{\sigma} \deq \{ (\sigma', \rev\eta):\ s \move[\eta\in\Eps] \sigma\con\sigma' \}.
$$
In particular, $\C{s} = \P{s}{\epsilon}$.
By definition, the derivatives of any $\P{s}{\sigma}$ can be expressed as
$$
  \P{s}{\sigma}\deriv{u} = \{ (\epsilon, \P{s}{\sigma \con u}) \}.
$$
This, however, does not directly lead to a \emph{finite} family closed under derivatives.
Instead, we show the following.

\begin{lemma}\label{lem:atomic-regular}
  For every $s, t, u \in S$, the derivatives of $\P{s}{t}$ are
  $$
    \P{s}{t}\deriv{u} = \left\{ \left( \bigcup_{d = \call{t',t'',u} \in D} \D{t''}{t} \con d,\ \P{s}{t'} \right) \right\}_{t' \in S}.
  $$
\end{lemma}
\begin{proof}
  The configurations in $\P{s}{t}$ are those for which $s \move[\Eps] t\con\sigma$.
  Taking derivative by $u$, we are interested in $\sigma = u \con \sigma'$ for some $\sigma'$.
  Therefore, consider a sequence of transitions $\eta\in\Eps$ which realizes $s \move[\eta] t \con u \con \sigma'$
  and split it at the point where $u \con \sigma'$ becomes the suffix of the configuration, not to be touched later.

  What lies above $u \con \sigma'$ cannot be empty (as then $u$ would get touched),
  which means that $u$ must be the lower state placed on the configration by some call.
  Thus the sequence $\eta$ looks as follows:
  $$
    s \move[\eta_1] t' \con \sigma' \move[\call{t',t'',u}] t'' \con u \con \sigma' \move[\eta_2] t \con u \con \sigma',
  $$
  with $\rev{\eta_2} \in \D{t''}{t}$ and $(\sigma', \rev{\eta_1}) \in \P{s}{t'}$, as required.
\end{proof}

\begin{corollary}
  \label{lem:closure-regular}
  For each $s \in S$, the language and relation $\C{s}$ are $(|S|+1,|S|)$-regular.
\end{corollary}
\begin{proof}
  The definitions of $\C{s}$, derivatives, and regularity for relations
  degenerate cleanly to those for languages when specific transitions are forgotten.
  Therefore we only need to prove the statement about relations.
  Indeed, by the above lemma,
  for each $s \in S$, the family containing $\C{s} = \P{s}{\epsilon}$ and $\P{s}{t}$ for every $t \in S$
  is closed under derivatives of degree $|S|$.
\end{proof}

This result can be translated to the null closures of atomic languages and relations:
when taking derivatives, nullable states might need to be skipped over
(with their nulling transitions given in the first component of the $\Lambda$),
which in the language of finite automata corresponds to taking $\epsilon$-closure,
possibly for the price of increased nondeterminism.

\begin{proposition}
  For each $s \in S$, the guided language and relation $\N{\C{s}}$ are $(|S|+1, |S|)$-regular.
\end{proposition}
\begin{proof}
  Again, only the statement about relations needs to be proven.

  The required family closed under derivatives are simply the null closures of those given above, i.e., $\N{\P{s}{\epsilon}}$ and $\N{\P{s}{t}}$.
  We would like to show that for some $\Lambda_{u,t,t'} \subseteq D^*{\times}D^*$ we have
  $$
    \N{\P{s}{t}}\deriv{u} = \left\{\left( \Lambda_{u,t,t'},\ \N{\P{s}{t'}} \right)\right\}_{t' \in S}.
  $$

  To express the transition relations $\Lambda_{u,t,t'}$ which are to be split off when taking derivatives,
  let us first give a name to the transition languages from Lemma~\ref{lem:atomic-regular}, denoting
  $$
    \Delta_{u,t,t'} \deq \bigcup_{d = \call{t',t'',u} \in D} \D{t''}{t} \con d.
  $$

  Now consider calculating $\N{\P{s}{t}}\deriv{u}$, fix some $t' \in S$
  and imagine a pair $(\gamma, \rev\eta) \in \N{\P{s}{t}}$ which will participate in this component of the derivative.

  The configuration $\gamma$ might begin with the state $u$.
  In this case, we can simply take the derivative $\P{s}{t}\deriv{u}$ and null-close the result.
  Therefore, $\Lambda_{u,t,t'}$ must contain $\epsilon {\times} \Delta_{u,t,t'}$.

  If $\gamma$ does not start with $u$, it must begin by a sequence of guiding transitions $\eta'$ which nulls some state $u'$.
  The state must have been present on top of the configuration in $\P{s}{t}$ to whose null closure $\gamma$ belongs.
  Therefore, we might take $\P{s}{t}\deriv{u'}$, moving to $\P{s}{t''}$ (note that the state $t''$ might be different than $t'$),
  and then recursively calculate the derivative $\N{\P{s}{t''}}\deriv{u}$.
  Therefore, $\Lambda_{u,t,t'}$ must also contain $\rev{\D{u'}{\epsilon}}{\times}\Delta_{u',t,t''} \ \con\ \Lambda_{u,t'',t'}$.

  As there are no other possibilities, we conclude that the relations $\Lambda_{u,t,t'}$ must be the least (wrt. inclusion) solutions to the equations
  $$
    \Lambda_{u,t,t'} = (\epsilon {\times} \Delta_{u,t,t'}) \ \cup\ \bigcup_{u',t'' \in S} \left(\rev{\D{u'}{\epsilon}} {\times} \Delta_{u',t,t''}\right) \con \Lambda_{u,t'',t'}.
  $$

  For the derivative $\N{\C{s}}\deriv{u} = \N{\P{s}{\epsilon}}\deriv{u}$, the reasoning is similar:
  either we go directly to the target relation $\N{\P{s}{u}}$, splitting off nothing (as was the case in Lemma~\ref{lem:atomic-regular}),
  or we skip over a nullable state $u'$ by taking derivative by it (this time this can only move us to $\P{s}{u'}$)
  and recursively calculate $\N{\P{s}{u'}}\deriv{u}$.
  We conclude that
  $$
    \N{\C{s}}\deriv{u} = \{ (\epsilon, \N{\P{s}{u}}) \} \ \cup\ \left\{\left( \bigcup_{u' \in S} \left(\rev{\D{u'}{\epsilon}}{\times}\epsilon\right) \con \Lambda_{u,u',t},\ \N{\P{s}{t}} \right)\right\}_{t \in S},
  $$
  where the singleton case on the left might be incorporated into the family on the right by including the $(\epsilon,\epsilon)$ pair in the transition relation for $t = u$.
\end{proof}

As a consequence,
once the necessary transition relations and languages have been pre-computed (this might happen at parser generation time),
the closures of grammar states, together with all their derivatives,
might be represented during parsing simply by pairs of states.

We also need to extract the transition languages corresponging to the empty configuration---%
or, for guided relations, to the configurations containing only guiding transitions and no states.
The first task is simple, because by definition
$$
  (\epsilon, \rev\eta) \in \P{s}{\sigma} \iff \rev\eta \in \D{s}{\sigma}.
$$
In the guided case, we might again need to skip over some nullable states.
Denoting
$$
  \L{s}{\sigma} \deq \P{s}{\sigma} \cap (D^*{\times}D^*),
$$
by an argument analogous to the one above we can conclude that
\begin{align*}
  \L{s}{t} &= (\epsilon{\times}\D{s}{t}) \cup \bigcup_{u',t'' \in S} \left( \rev{\D{u'}{\epsilon}} {\times} \Delta_{u',t,t''} \right) \con \L{s}{t''} \text{ and} \\
  \L{s}{\epsilon} &= (\epsilon{\times}\D{s}{\epsilon}) \cup \bigcup_{u' \in S} \left( \rev{\D{u'}{\epsilon}}{\times}\epsilon \right) \con \L{s}{u'}.
\end{align*}

The equations given for $\Lambda_{u,t,t'}$ and $\L{s}{t}$ follow directly from the above reasoning,
but using them as they stand would make our algorithm potentially exponential in the RTN size.
Therefore, we rewrite them into an equivalent form.

Let us begin with the original equation, with indices renamed:
$$
  \Lambda_{u_0,t_0,t'}
  = (\epsilon{\times}\Delta_{u_0,t_0,t'}) \ \cup\ \bigcup_{u_1,t_1 \in S} \left( \rev{\D{u_1}{\epsilon}} {\times} \Delta_{u_1,t_0,t_1} \right) \con \Lambda_{u_0,t_1,t'}.
$$
We unwind the recursion to an arbitrary depth $r$, obtaining
$$
  \Lambda_{u_0,t_0,t'}
  = \bigcup_{\substack{r \geqslant 0, \\ u_i,t_i \in S \\ \text{ for } 1 \leqslant i \leqslant r}}
    \left( \rev{\D{u_1}{\epsilon}} {\times} \Delta_{u_1,t_0,t_1}  \right) \cons
    \left( \rev{\D{u_r}{\epsilon}} {\times} \Delta_{u_r,t_{r-1},t_r}  \right) \con
    (\epsilon {\times} \Delta_{u_0,t_r,t'}),
$$
then split the cases for $r=0$ and $r>0$, and in the latter regroup the final terms and repack the rest using the equation above:
$$
  \Lambda_{u_0,t_0,t'}
  = (\epsilon {\times} \Delta_{u_0,t_0,t'}) \ \cup\  \bigcup_{u_r,t_r \in S} \Lambda_{u_r,t_0,t_r} \con \left( \rev{\D{u_r}{\epsilon}} {\times} \Delta_{u_0,t_r,t'} \right).
$$

Analogously, the equations for $\L{s}{t}$ are replaced by the equivalent
$$
  \L{s}{t_0} = (\epsilon {\times} \D{s}{t_0}) \ \cup\  \bigcup_{u_r,t_r \in S} \Lambda_{u_r,t_0,t_r} \con \left( \rev{\D{u_r}{\epsilon}} {\times} \D{s}{t_r} \right).
$$

\subsection{Transition relations and languages}
\label{sec:transition-api}

Moving down the ladder of abstraction,
it is time to find the representations $\type\Lambda$ and $\type\Delta$ for transition relations and languages.

For $\type\Lambda$, note that the only way its values,
formed from $\epsilon$ by chains of calls to \texttt{prepend} and \texttt{union},
can ever be examined is by calls to \texttt{flatten}.
Consider first a value formed using only \texttt{prepend}, i.e., of the form
\[
  \Lambda = \Delta_1 \con \Lambda^\circ_1 \cons \Delta_k \con \Lambda^\circ_k.
\]
The flattening of such value can be expressed as
\[
  \flat\Lambda = \flat{ \ldots \flat{\Delta_1 \con \Lambda^\circ_1} \ldots \con \Delta_k \con \Lambda^\circ_k}.
\]
To perform this computation, it is enough to require the type $\type\Lambda^\circ$ to support a strengthened version of flattening,
which we call \bld{wrapping}:
\begin{itemize}
  \item $\texttt{wrap}(\Delta \in \type\Delta, \Lambda^\circ \in \type\Lambda^\circ) \in \type\Delta$ \tab $\flat{\Delta \con \Lambda^\circ}$.
\end{itemize}

If forming the $\Lambda$ in question involved some calls to \texttt{union},
we will make use of the fact that flattening distributes over union.
The whole expression for $\Lambda$ can be seen as a DAG with a single distinguished sink node, and edges labeled with calls to \texttt{prepend}.
Each path through this DAG represents a union-free expression as discussed above.
The relation represented by a node will be the union over all paths from this node to the sink.
Shall \texttt{flatten} be requested for a node $v$ representing some $\Lambda_v$,
the value can be computed following a topological traversal from $v$,
in time linear in the size of the DAG, and thus in the number of operations used to form $\Lambda_v$.

For the atomic relations (i.e., the type $\type\Lambda^\circ$),
we "represent" them directly by the implementation of their \texttt{wrap} operation.
We only need to express how the result of \texttt{wrap} (for an arbitrary $\Delta$) will look like
for the primitives returned by API~\ref{api:atomic-relations}, namely
$$
  \L{s}{\epsilon}, \
  \L{s}{t}, \
  \Lambda_{u,t,t'}\ \text{and}\
  \bigcup_{u' \in S} \left(\rev{\D{u'}{\epsilon}}{\times}\epsilon\right) \con \Lambda_{u,u',t}.
$$

This can be given by the following set of recursive equations:
\begin{align*}
  \flat{\Delta \con \Lambda_{u,t,t'}} &=
    \Delta \con \Delta_{u,t,t'} \cup \bigcup_{u',t'' \in S} \D{u'}{\epsilon} \con \flat{\Delta \con \Lambda_{u',t,t''}} \con \Delta_{u,t'',t'}, \\
  \flat{\Delta \con \left( \bigcup_{u' \in S}\left( \rev{\D{u'}{\epsilon}}{\times}\epsilon \right) \con \Lambda_{u,u',t} \right)} &=
    \bigcup_{u' \in S} \flat{\D{u'}{\epsilon} \con \Delta \con \Lambda_{u,u',t}}, \\
  \flat{\Delta \con \L{s}{t}} &=
    \Delta\con\D{s}{t} \cup \bigcup_{u',t'' \in S} \D{u'}{\epsilon} \con \flat{\Delta \con \Lambda_{u',t,t''}} \con \D{s}{t''}, \\
  \flat{\Delta \con \L{s}{\epsilon}} &=
    \Delta\con\D{s}{\epsilon} \cup \bigcup_{u' \in S} \flat{\D{u'}{\epsilon} \con \Delta \con \L{s}{u'}}.
\end{align*}
Of course, it makes sense to permanently remove from the calculation of the above unions those components which are never going to be relevant,
i.e., those for which $\D{u'}{\epsilon}$ or $\Delta_{u',t,t''}$ is empty.

\begin{api}\label{api:transition-languages}
  Required operations on transition languages (subsets of $D^*$)
  \begin{itemize}
    \item $\texttt{epsilon}() \in \type\Delta$ \tab $\{\epsilon\}$,
    \item $\texttt{concatenate}(\Delta_1, \Delta_2 \in \type\Delta) \in \type\Delta$ \tab $\Delta_1 \con \Delta_2$,
    \item $\texttt{union}(\Delta_1, \ldots, \Delta_k \in \type\Delta) \in \type\Delta$ \tab $\Delta_1 \cup \cdots \cup \Delta_k$,
    \item $\texttt{is-empty}(\Delta \in \type\Delta) \in \{\texttt{true},\texttt{false}\}$ \tab whether $\Delta = \varnothing$,
    \item $\texttt{primitive}(?) \in \type\Delta$ \tab (see text)
  \end{itemize}
\end{api}

As the final piece of data representation, we turn to the type $\type\Delta$ encoding transition languages.
The basic operations performed on its elements during parsing are union, concatenation, and checking (non)emptyness (cf.~API~\ref{api:transition-languages}).
However, one also needs a rather rich set of "primitives" (which could best be pre-computed during parser generation):
singletons of individual shift transitions, $\D{s}{\epsilon}$, $\D{s}{t}$ and $\Delta_{u,t,t'}$,
and, moreover, the ability to solve equations of the shape given in the above discussion concerning $\type\Lambda$.
We employ the equation solver (to be discussed shortly) to compute the primitives by expressing
\begin{align*}
  \D{s}{\epsilon} &= \bigcup_{d = \reduce{s,p} \in D} d \ \cup\ \bigcup_{d = \call{s,u,t} \in D} \D{t}{\epsilon} \con \D{u}{\epsilon} \con d, \\
  \D{s}{t} &= \{ \epsilon:\ s = t \} \ \cup\ \bigcup_{d = \call{s,u,s'} \in D} \D{s'}{t} \con \D{u}{\epsilon} \con d, \quad\text{and} \\
  \Delta_{u,t,t'} &= \bigcup_{d = \call{t',t'',u} \in D} \D{t''}{t} \con d.
\end{align*}
This way, we are only left with individual transitions, concatenation, union, and equations expressible in terms of these.

To assess the required capabilities of the equation solver, we will analyze the dependency graph between variables.
In particular, when this graph turns out to be acyclic, the solutions might be calculated simply by evaluation.
This time, we start with the variables at the end of the dependency chain.

\bld{$\D{s}{\epsilon}$.}
Every $\call{s,u,t} \in D$ seems to induce a dependency between $\D{s}{\epsilon}$ and each of $\D{u}{\epsilon}$ and $\D{t}{\epsilon}$.
However, this dependency is only real if $\D{t}{\epsilon}$ (respectively, $\D{u}{\epsilon}$) is nonempty,
which is exactly when $t$ (respectively, $u$) is nullable.
Let us then assume that nullability has been computed by some other means (see, e.g.,~\cite{cyk-2nf}, Section 4.2),
and announce a \bld{right dependency} $\D{s}{\epsilon} \dep \D{t}{\epsilon}$ if $\call{s,u,t} \in D$ for some nullable $u$,
and a \bld{left-right dependency} $\D{s}{\epsilon} \dep \D{u}{\epsilon}$ if $\call{s,u,t}\in D$ for some nullable $t$.
The direction(s) of dependencies record on which side of their targets a nonempty sequence of transitions might occur.

\bld{$\D{s}{t}$.}
Dependencies on $\D{s}{\epsilon}$ are irrelevant, as they will never introduce cycles.
The only internal dependencies in this layer are
a right dependency $\D{s}{t} \dep \D{s'}{t}$ if $\call{s,u,s'} \in D$ for some nullable $u$.
Note that it occurs in accord with an identical dependency $\D{s}{\epsilon} \dep \D{s'}{\epsilon}$.

\bld{$\flat{\Delta \con \Lambda_{u,t,t'}}$.}
The dependency graph between these variables is independent of $\Delta$
(unless $\Delta = \varnothing$, in which case the solution is trivially $\varnothing$)
and also of $t$,
let us then omit them from node descriptions.
We have a left-right dependency $\Lambda_{u,t'} \dep \Lambda_{u',t''}$ if and only if
$u'$ is nullable and $\Delta_{u,t'',t'} \neq \varnothing$, which in turns happens if we have a $\call{t',t''',u} \in D$ with $\D{t'''}{t''} \neq \varnothing$.
Note, that a left-right dependency $\D{s}{\epsilon} \dep \D{u}{\epsilon}$, witnessed by a nullable $t$ as discussed above, implies $\Lambda_{t,s} \dep \Lambda_{t,u}$.

Other variables introduce no potentially cyclic dependencies.
This means there are only two ways in which a cycle can appear, making solution-by-evaluation impossible:
\begin{enumerate}
  \item
    There is a cycle of left-right dependencies between the nodes denoted $\Lambda_{u,t}$.
    Consider such a cycle $\Lambda_{u_1,t_1} \dep \Lambda_{u_2,t_2} \dep \ldots \dep \Lambda_{u_r,t_r} = \Lambda_{u_0,t_1}$,
    and imagine a successful run of the RTN in which the state $t_1$ appears.
    The history of any such run must look as follows:
    $$
      \sstart \move[\delta_1] t_1\con\sigma \move[\delta_2] \sigma \move[\delta_3] \epsilon
    $$
    (note that this time $\delta_1$, $\delta_2$ and $\delta_3$ might contain shifts).
    Now, for each $i$, the dependency $\Lambda_{u_i,t_i} \dep \Lambda_{u_{i+1},t_{i+1}}$ is witnessed by $t_i \move[\eta_i] t_{i+1} \con u_i$.
    At the same time, each $u_i$ is guaranteed to be nullable,
    which implies that the original run could be replaced with an equally successful
    \begin{align*}
      \sstart &\move[\delta_1] t_1\con\sigma \move[\eta_1] t_2 \con u_1 \con \sigma \move[\eta_2] \ldots \move[\eta_{r-1}] t_0 \con u_{r-1} \cons u_1 \con \sigma \move[\delta_2] \\
      &\move[\delta_2] u_{r-1} \cons u_1 \con \sigma \move[\Eps] \sigma \move[\delta_3] \epsilon.
    \end{align*}

    Iterating this trick, one obtains an infinite set of accepting transition sequences for the same input.
    Moreover, as the insertions on the left and right of $\delta_2$ need to be balanced (this is the reason we call such dependencies left-right),
    the language of those accepting runs is not regular (it is obviously context-free).

    Together, either every state $t_i$ on such cycle is useless and could be removed, or the RTN is \emph{nonregularly infinitely ambiguous}.
    One could argue that such grammars could be immediately rejected during parser generation%
    \footnote{That is what our proof-of-concept implementation does.}
    (the check for this situation is fairly simple).
    If one decides to handle them, the type $\type\Delta$ needs to be equipped with some mechanism for handling context-free (nonregular) structures.
  \item
    There is a cycle of joint right and left-right dependencies between the nodes $\D{s}{\epsilon}$.
    Again, one could argue that either the whole such cycle is useless, or the RTN is infinitely ambiguous.
    However, if all such cycles consist of right dependencies only, the situation is somewhat less complicated:
    mutual dependencies exist \emph{only} between the variables $\D{s}{\epsilon}$, and all those dependencies are left-linear.
    In such case, the languages in question are regular (albeit still infinite) and the only additional feature required of the type $\type\Delta$ is the Kleene star operator.
\end{enumerate}

From the discussion above we can see that, unless one needs to handle infinitely ambiguous grammars,
the requirements on $\type\Delta$ are quite small:
the \texttt{primitive} operaton in API~\ref{api:transition-languages} only needs to handle individual transitions.

This leaves a wide spectrum of implementation possibilities.
As mentioned previously, one can instantiate $\type\Delta$ to $\type{B}$, obtaining a recognition algorithm.
On the other extreme, one can use a free algebra over \texttt{concatenate} and \texttt{union}---%
the result will be an encoding of all parse trees, very similar in spirit to the SPPF structure employed by polynomial-time generalized parsers.
Both of the above require $O(1)$ time per operation on $\type\Delta$.

There are, however, other options.
In particular, one may keep singleton languages only, replacing larger ones by a unique ambiguity marker.
The result is a de-generalized parser, finding the unique parse tree for its input (or reporting ambiguity).
If the languages are remembered up to (and including) the first ambiguity, the location and nature of the ambiguity can be reported to the user.

Further, one can introduce a partial ordering of transitions,
and "prune" each intermediate language to contain only lexicographically-maximal transition sequences.
This achieves transition prioritization,
useful among others for expressing operator associativity and precedence,
and eager/lazy versions of repetition constructs.
All of this happens without any modification to the parsing algorithm itself!

\section{Arbitrary semiring parsing}
\label{sec:semiring}

If we restrict our attention to finitely ambiguous RTNs,
the API we require of transition languages $\type\Delta$ is simply the one of a semiring.
In fact, the alternatives we sketched in Section~\ref{sec:transition-api} are just special semirings.
Substituting $\type\Delta$ with an arbitrary semiring encounters, however, a hidden limitation---%
the algorithm, as presented above, requires semiring addition to be idempotent.
Lifting this limitation would allow faster calculation of other useful features (cf.~\cite{semiring}), e.g., the number of distinct parse trees.
We will now show how it can be achieved.

Let us then assume that $(\type\Delta, {\cup}, \varnothing, \con, \epsilon)$ is an arbitrary semiring,
and that a \bld{valuation function} $\val\relbar : D \to \type\Delta$ assigns values to individual transitions.
Valuation can be naturally extended to a semiring homomorphism
between the set of languages over $D$ (i.e., subsets of $D^*$) and the semiring $\type\Delta$:
$$
  \val\Delta \deq \bigcup_{d_1 \ldots d_k \in \Delta} \val{d_1 \rangle \cons \langle d_k}.
$$

We would like our parsing algorithm to return, for an input $\alpha$, the value
$$
  \val{\{ \delta\con\eta: (\eta,\rev\delta) \in \N{\Phi_\alpha} \}}.
$$
Valuation can be pre-computed for the "building blocks" of the languages of transitions
(individual transitions, nullings of individual states, and whatever is needed for the derivatives of $\N{\C{s}}$),
and then, being a homomorphism, pushed through the algorithm to its final result.
It might happen, however, that at some point of the computation,
the same sequence of transitions appears as a member of both arguments of a union operation---%
we will say that such computation is \bld{redundant}.
If $\cup$ is idempotent (as set union is), this does not hurt. If it is not, such a sequence will be "counted twice".
Therefore, to make the algorithm work correctly over an arbitrary semiring, we must assure that such a situation never takes place.

Redundancy may be potentially introduced wherever a union of transition languages appears in the algorithm, namely when:
\begin{itemize}
  \item (pre)computing the valuation of transition relations and languages needed by the derivatives of $\N{\C{s}}$,
  \item computing derivatives of degree $k$---if the same pair appears in more than one of the $k$ components, or
  \item composing $\Psi\out$ or $\Psi_\texttt{aux}$ in Algorithm~\ref{alg:phase-parse-final}.
\end{itemize}
The first case we have already avoided---%
in the equations considered in Section~\ref{sec:transition-api},
the sets of transition sequences in each member of each union are guaranteed to be disjoint.
The second is an additional requirement to be put on the API for guided relations---%
a requirement which the representation proposed in Section~\ref{sec:dag} also already satisfies.
The remaining part of this section addresses the third case.

As it stands, Algorithm~\ref{alg:phase-parse-deep} (and thus the equivalent Algorithm~\ref{alg:phase-parse-final}) is redundant in two ways.
Firstly, some sequences of transitions (and their corresponding configurations) may be added to $\Psi_\texttt{out}$ both in lines~4 and~7.
Secondly, when it is used iteratively by Algorithm~\ref{alg:main-parse},
some sequences of transitions have multiple ways in which they can be split between successive iterations.

It turns out that a careful adjustment of the nulled atomic closure relations $\N{\C{s}}$ is enough to completely get rid of redundancy---%
neither the top-level parsing algorithm, nor the API (or implementation) of guided relations $\type\Psi$ needs to be modified.
Formally, we change the meaning of call-reduce closure never to include the empty configuration:
$$
  \C{\Psi} \deq \{(\gamma', \eta\con\delta) :\ \gamma \move[\rev\eta \in \Eps] \gamma' \neq \epsilon, (\gamma,\delta) \in \Psi \},
$$
and restrict null closure not to touch the top of the configuration:
$$
  \gamma_1 \con s \con \gamma_2 \in \Psi, \quad \gamma_1 \neq \epsilon \quad \implies \quad \gamma_1 \con \N{s} \con \gamma_2 \subseteq \N\Psi.
$$

This modification silently alters the semantics of both the main algorithm and phase calculation.
We should then prove that they remain correct.

\todo[inline]{TODO: Re-prove Propositions~\ref{prop:main-parse} and~\ref{lem:phase-parse}}

Finally, we need to show that the atomic closure relations remain regular.
However, the only change we have effectively made to them is requiring each guided configuration to begin with a state.
Therefore, the nondeterministic automaton witnessing regularity can be obtained by:
\begin{itemize}
  \item making the initial state non-accepting: $$\N{\C{s}} \cap (D^*{\times}D^*) = \varnothing,$$
  \item using the non-guided derivatives from the initial state: $$\N{\C{s}}\deriv{u} = \{ (\epsilon, \N{\P{s}{u}} ) \}, \text{ and}$$
  \item leaving everything else as in the guided variant.
\end{itemize}

Thus, getting rid of redundancy not only generalizes relational parsing to arbitrary semirings,
but has also an additional positive effect of making the implementation simpler and more effective,
as fewer edges are created in the dag structure for $\type\Lambda$.

\section{Complexity}
\label{sec:complexity}

We will now show that the input-related complexity of our algorithm from Sections~\ref{sec:recognition} and~\ref{sec:parsing},
run using the representations of languages given in Section~\ref{sec:representations},
matches that of Earley parser~\cite{earley} with Leo's improvement~\cite{leo}.
To our best knowledge this is the best among state-of-the-art generalized parsers.

\begin{proposition}
  \label{prop:complexity-cubic}
  With the proposed implementation, a sequence of $n$ "constructive" operations from API~\ref{api:relations} and a single call to \textup{\texttt{epsilon}} takes time $O(n^3)$.
\end{proposition}
\begin{proof}
  As noted before, each language (relation) in such a sequence can be represented by at most $n\,|\type\Xi| = n(|S|+1) = O(n)$ pairs.
  When computing \texttt{derivative}, one possibly needs to examine (and combine) the representations of all $\Psi_{i_j}$ "directly referenced" by the current $\Psi_n$.
  This takes time $O(n^2)$.
  Running times of \texttt{union} and \texttt{prepend} are linear in the size of their arguments.
  Summing over the sequence of $n$ operations gives the desired $O(n^3)$.
  Returning \texttt{epsilon} involves flattening a single label, and thus contributes only a constant additional factor to the total time.
\end{proof}

\begin{corollary}
  Parsing an input of length $n$ takes time $O(n^3)$.
\end{corollary}

To be able to give better complexity bounds for restricted classes of grammars,
we will make the standard natural assumption that each state in the grammar is \bld{productive},
i.e., yields some (possibly empty) terminal word.
Detecting and removing non-productive states can be easily done during parser generation.

The consequence of this assumption is that every reachable \emph{configuration} is also productive.
This allows us to prove the following:

\begin{proposition}
  \label{prop:complexity-quadratic}
  When the recursive transition network is unambiguous,
  parsing an input of length $n$ takes time $O(n^2)$.
\end{proposition}
\begin{proof}
  Each way a configuration can "enter" into a parsing language corresponds to some viable sequence of RTN transitions.
  In particular, should any configuration have more than one such way, the grammar would be ambiguous.
  This implies that, for unambiguous grammars, all unions---%
  those requested explicitly and those performed internally by \texttt{derivative}, are in fact disjoint.
  This makes each operation run in time proportional to the size of its result (i.e., the number of edges of the vertex created).
  Consequently, the complexity of the \emph{whole sequence} of operations is linear in the size of the complete DAG constructed.
  The result follows.
\end{proof}

It turns out that for LR-regular grammars (this class contains LR(k) for all k, and thus most deterministic grammars),
the size of the DAG is always linear in the input length:

\begin{proposition}
  \label{prop:complexity-linear}
  When parsing according to a recursive transition network which is LR-regular,
  the outdegree of each DAG vertex created is bounded by a constant.
\end{proposition}
\begin{proof}
  Let us classify each vertex of the DAG depending on the kind of operation which has created it.
  Note that \texttt{prepend} with an $\epsilon$-free language always creates a single edge to its argument,
  while \texttt{union} and \texttt{derivative} only copy existing edges, with labels possibly altered.
  Analyzing our phase computation (Algorithm~\ref{alg:phase-recognize-deep}),
  one can see that only \texttt{derivative} vertices will ever have incoming edges,
  and that only edges going out of such vertices need to be considered---%
  \texttt{prepend} vertices always have outdegree 1,
  while \texttt{union} vertices combine a constant number of \texttt{prepend} vertices.

  Denote by $\Sigma_i$ the \texttt{union} vertex representing the parsing language computed in the $i$-th phase of the algorithm.
  Consider a vertex $\Sigma'$ whose first incoming edge is created by one of the members of such \texttt{union},
  and let $\N{\C{u}}$ be the label of this edge.
  This implies that $\Sigma_i \supseteq \efree{\N{\C{u}}} \con \Sigma'$.

  Now track the progress of new incoming edges being added to $\Sigma'$,
  resulting with it being linked from either $\Sigma_j\deriv{s}$ or $\Sigma_j\deriv{s \con s'}$ for some $j \geqslant i$ and $s, s' \in S$.
  The labels must have been obtained by taking successive derivatives of the original one.
  Let $\tau$ denote this sequence of derivatives.

  In the first of the cases above, it implies that $\efree{\N{\C{u}}\deriv{\tau}} \con \Sigma' \subseteq \Sigma_j\deriv{s}$.
  As our DAG contains no useless edges, this inclusion must not be trivial:
  there must exist $\tau \con s \con \sigma \in \N{\C{u}}$ and $\sigma' \in \Sigma'$ for which $s \con \sigma \con \sigma' \in \Sigma_j$.
  Looking at the sequence of transitions witnessing the first of these three conditions,
  and splitting it immediately before $s$ enters the configuration,
  one can find $v \con \sigma \con \sigma' \in \Sigma_i$ and $\call{v,w,s} \in D$
  such that $w$ generates the part of the input between positions $i$ and $j$.
  Yet, by an argument of Leo (\cite{leo}, Lemma 4.7),
  for LR-regular grammars each $j$ allows only a constant number of different $i$'s where this is possible.

  A similar reasoning can be carried over for the vertices $\Sigma_j\deriv{s \con s'}$.
  Together with the fact that $O(1)$ vertices are created for each input position,
  one obtains the desired constant bound on the out-degree of every DAG vertex.
\end{proof}

\begin{corollary}
  When the recursive transition network is LR-regular,
  parsing an input of length $n$ takes time $O(n)$.
\end{corollary}

\section{Memoization}
\label{sec:memoization}

We find the mathematical clarity of our algorithm to be a virtue on its own.
In particular, it opens a clear way of investigating potentially more effective data representations,
and making use of language-theoretic algorithms, like optimization of nondeterministic automata.
However, an aspect we find especially beneficial is the immutability of the intermediate objects computed during parsing.
This is because it allows memoizing and reusing the results of computation.
We report on this aspect in this section, beginning with recognition and moving on to arbitrary semiring parsing.
Experimental evaluation of all these techniques is given in Section~\ref{sec:evaluation}.

\subsection{Trivial memoization}

The most obvious way in which one can reuse the results of previous phases of computation is the following:
if $\N{\Sigma_\alpha} = \N{\Sigma_\beta}$ then $\N{\Sigma_{\alpha \con a}} = \N{\Sigma_{\beta \con a}}$.
To build an intuition on when configuration languages at different input positions can be equal,
imagine the structure of a standard, "hierarchical" programming language,
in which syntactic structures (like classes, functions, code blocks, parenthesized expressions, etc.) nest within one another.
A grammar of such a language will most likely contain productions corresponding to particular structures.
In turn, a parser configuration at an input position where some structures are "currently open"
will contain a sequence of states from those productions.
Of course, even for deterministic grammars, there can be more than one reachable configuration at a given point.
However, the general intuition holds: input locations with identical "nesting structure" will often share configuration languages.
This happens both within a single input, and across multiple inputs.

\subsection{Dominator-based memoization}

To further develop our inuition on configuration languages,
imagine a syntactic language structure requiring a specific closing terminal (e.g., \texttt{\}}).
Looking at the parsing stack, the state (say, $s$) corresponding to such structure being "currently open" cannot be reduced until the closing terminal is encountered.
During this period, every configuration will have the form $\sigma \con s \con \tau$, with varying $\sigma$ but identical $\tau$.
Our parser will not even look at $\tau$ until it is uncovered by taking a derivative by $s$.
This means that computation results for one $\tau$ can be "transplanted" on top of another,
as long as they have at least one state "covering" this bottom part of the stack.

To make use of this observation we need to significantly complicate our representation of languages.
A language will now be represented as a \emph{stack} of DAG vertices,
its semantics being the concatenation (top-to-bottom) of languages corresponding to individual vertices.
To keep the languages deeper in the stack "covered", we require all but the bottom vertices \emph{not} to be final (i.e., to denote $\epsilon$-free languages).
Should we ever attempt to push a final vertex on a nonempty stack, we instead merge it with the current top,
computing the concatenation of their corresponding languages.

During each phase, we keep track of how long a prefix of the stack we have seen.
Note, that uncovering $k$ deeper vertices requires taking derivatives by at least $k$ states, and a single phase takes derivatives by at most two of them.
Together with the potential merge to keep the $\epsilon$-free invariant, this means that a phase will examine at most three top stack entries.
When a phase is done, its result (a single DAG vertex $v$, about to be placed on top of the part of the stack which has not been touched) is first \bld{factorized}:
the fragment of the DAG between $v$ and all final vertices reachable from it is split at all their common dominators.
The resulting factors are represented by individual vertices, and these are pushed on the stack instead of $v$.

Such partial phase computations are held in a trie,
which maps a terminal and a prefix (with at most three elements) of the stack into a prefix it should get replaced with.
Whenever the top of the stack is found in the trie, phase computation requires only a handful of simple stack operations.

\subsection{Memoizing arbitrary semirings}

We have shown how to effectively use memoization in our \emph{recognizer}.
It is possible to use the same techniques for parsing (over an arbitrary semiring).
However, the quality of memoization drops rapidly as the semiring domain grows.
It is best exemplified in the "reachest", free semiring, distinguishing all transition languages.
In this case, each transition sequence occurring in the parsing relation $\Phi_\alpha$ encodes the parsed input $\alpha$
(the input can be read off the \texttt{shift} transitions).
This means that complete parsing relations \emph{never} repeat while parsing a single input.
The situation is not that bad with dominator-based memoization, but the efficiency is still unsatisfactory.

Effective memoization of full-fledged relations is made possible by combining two observations.
First, the shapes of the DAGs constructed while parsing are exactly the same as during recognition
\footnote{This assumes identical "shapes" of atomic closures.
When the grammar or atomic automata are optimized, these shapes might depend on the chosen semiring.
Still, shape repeatability is in practice very high regardless of that choice.}
and thus \emph{do} repeat very often---it is the labels that do not.
Second, the semiring API allows label values to be combined, but not examined:
the only exception is \texttt{is-empty}, but this information is also present in the shape of the DAG.

Imagine an instance of our labeled DAG structure
in which all labels have been replaced with variables of unknown (non-zero) values.
Our phase computation algorithm will happily accept such a DAG, creating some new labeled vertices and edges.
The values of these fresh labels can be specified using expressions in the above variables,
and such expressions depend only on the shape of the original DAG.
Thus, one may view the phase computation as a combination of two parts:
the first specifies how the shape of the output DAG depends on the shape of the input one,
while the second prescribes a substitution, allowing to compute output labels from input labels.

With the trivial memoization, most of the substitution will be trivial,
as a single phase only touches the fragment of the DAG close to its current top.
Here, the dominator-based approach not only improves memoization rates,
but also helps to avoid the trivial substitutions.

The complete representation of relations now looks as follows.
It is still a stack of DAG vertices, but instead of semiring values, the labels are given as expressions
with variables standing for labels of (up to three) stack nodes from which the current node has been obtained.
An additional pointer to this "origin" node makes it possible to perform the substitution when needed.
The trie used for memoizing actions is practically the same,
but only the shapes of the vertices (without labels) are taken into consideration.

Instead of a single semiring value,
the final result of parsing becomes now a fairly involved data structure,
describing an expression for this value as a hierarchy of small substitutions.
These substitutions have to be performed if actual value is to be returned,
and this computation often takes \emph{significantly longer} than parsing per se!
What have we saved then?
We claim that in order to arrive at the final value,
more or less the same semiring operations (possibly in a different order) would have to be performed by \emph{any} semiring-agnostic parser.
Our final substitution is then in some way optimal,
as it ends up computing exactly those intermediate values which are actually relevant for the result.

It is also interesting to see how the resulting substitution structure looks like for unambiguous languages.
If there is no legal parse of the input, the result is a single constant: the semiring zero.
Otherwise, all union operations in all substitutions must be degenerate, i.e., have exactly one operand.
Thus, the whole expression is a (complex) concatenation of individual transitions.
Each small expression given at each level of substitutions is then a sequence of transitions,
with "holes" to be filled with the values of referenced variables.
Viewed as a description of a parse tree, it is akin to a tree with "holes",
to be filled with subtrees obtained from referenced variables.
This suggests that for some applications one might not need to actually compute the final value,
but instead operate directly on the "compressed" substitution structure.
It is thus an interesting open question whether this way of encoding parse forests might be advantageous
over SPPFs returned by virtually all generalized parsers.

\section{Evaluation}
\label{sec:evaluation}

To evaluate the effects of memoization on practical programming languages,
we have chosen a grammar for Java 8 from the ANTLR4 grammar repository
(\url{https://github.com/antlr/grammars-v4/blob/master/java8/Java8.g4};
this grammar is ambiguous, but extremely close to the Java language specification)
and a large, popular Java project (\texttt{elasticsearch-6.5.3}, with 9596 \texttt{.java} files totalling over 10 million terminals).

The grammar, originally with 3132 RTN states, has been left-factored and then optimized by simple partition refinement,
analogous to Moore's method of minimizing finite automata.
At parser generation time, the automata representing atomic relations have been optimized in the same way.
All these optimizations depend on the semiring in which the parser is supposed to operate---
Figure~\ref{fig:automata-sizes} reports the numbers of states in all the cases.

\begin{figure}
  \centering
  \begin{tabular}{l|r|r|r|r}
    semiring & \multicolumn{2}{c}{RTN states} & \multicolumn{2}{c}{atomic states} \\
    & left-factored & optimized & original & optimized \\
    \hline\hline
    recognition & 2818 & 637 & 7816 & 2949 \\
    counting & 2818 & 652 & 7952 & 3694 \\
    parsing & 2818 & 637 & 7816 & 5627 \\
    \hline
  \end{tabular}
  \caption{Grammar and automata sizes under different semirings.}
  \label{fig:automata-sizes}
\end{figure}

We have performed the evaluation on a modest laptop, with 3 GHz Intel Core i7 CPU, running 64-bit Linux system.
To take into account the time required for garbage collection
(and for a more fair comparison with the \texttt{javac} compiler),
we have restricted the whole processes (including .NET Core or Java VM) to run on a single CPU core.
About 5 GB of RAM were available to the process at each run.

\begin{figure}
  \centering
  \begin{tabular}{lr|r|r|r}
    semiring & $|\type\Xi|$ & trivial & dominator & substitutions \\
    \hline\hline
    recognition & 2949 & 94.1 \% & 99.7 \% & 99.6 \% \\
    counting & 3694 & ! 49.3 \% & 97.2 \% & 99.6 \% \\
    parsing & 5627 & ! 7.5 \% & ! 14.4 \% & 99.6 \% \\
    \hline
  \end{tabular}
  \caption{Effectiveness of different memoization strategies.}
  \label{fig:memoization}
\end{figure}

Figure~\ref{fig:memoization} compares the effectiveness of different methods of memoization.
Each entry in this table gives the percentage of phases whose result was taken from the cache.
The mark ! denotes situations in which the algorithm would require too much RAM to complete---%
the value in such cases is approximate, since it only accounts for the phases performed before exhausting all available memory.

As expected according to the discussion in Section~\ref{sec:memoization},
trivial memoization quickly becomes ineffective as the semiring grows richer.
For dominator-based memoization using actual semiring values, the trend is similar, although not that rapid.
Memoization using substitutions may only depend on the semiring indirectly,
through the shape of the optimized atomic automata.
In our tests, they had almost no influence on its effectiveness.

With over 10 million parsing phases to be performed,
it is not by accident that only the variants with high ($> 90 \%$) memoization rate were able to complete---%
attempting to store millions of configurations is simply impractical.
Furthermore, with trivial memoization, performing the 5.9\% non-memoized phases took 96\% of total recognition time---%
this shows how expensive the DAG operations are compared to a simple dictionary lookup.
With dominator-based memoization, due to yet increased complexity of the data structures,
the 0.3\% non-memoized phases took about 40\% of total recognition time.
This shows that variants for which memoization is anything below excellent
would remain impractical even if given unlimited memory.

We will now share some raw performance comparison,
not only between variants of our algorithm,
but also with our implementation of GLL,
and the hand-coded deterministic parser built into the Java 8 compiler \texttt{javac}.
While our proof-of-concept implementation has not been created with speed in mind,%
\footnote{It has multiple additional layers of abstraction, useful for experimenting with different representations.
In C$\sharp$, the overhead of these layers cannot be eliminated at compile-time and must be paid during parsing.
It seems that a well-structured C++ implementation could do significantly better.}
the results are already a strong evidence of its practicality.

\begin{figure}
  \centering
  \begin{tabular}{l|r|r}
    algorithm & time & memory \\
    \hline\hline
    GLL & 2445.6 s & 0.32 GB \\
    GLL (non-optimized) & 470.2 s & 0.15 GB \\
    \hline
    relational (no memoization) & 751.4 s & 2.42 GB \\
    relational (trivial memoization) & 77.5 s & 3.22 GB \\
    relational (dominator-based memoization) & 7.8 s & 0.22 GB \\
    relational (dominator-based, non-optimized) & 30.3 s & 1.08 GB \\
    \hline
    \texttt{javac} (parser) & 14.0 s & 0.93 GB \\
    \hline
  \end{tabular}
  \caption{Running times and memory usage of compared recognition algorithms.}
  \label{fig:timings-recognition}
\end{figure}

Figure~\ref{fig:timings-recognition} gives time and memory usage of different recognizers.
Note that, while other timings are for parsing alone,
the one for \texttt{javac} includes tokenizing the input and building parse trees.
However, in a hand-coded, deterministic parser, we would not expect these two tasks to take significantly more than half of total execution time.
The comparison suggests that, with our method,
an unoptimized recognizer, automatically generated from an ambiguous language specification
can rival optimized, dedicated, hand-written, deterministic code!

The table includes two rows marked as "non-optimized".
The timings given there are for the appropriate algorithms run on the grammar which has only been left-factored,
and has not undergone partition refinement optimization.
As evident, while such optimization has clearly helped our approach (by reducing the cardinality of $\type\Upsilon$ from 5503 to 2949),
it has made GLL perform significantly worse (by making its GSS denser).
Not also, that if both algorithms are given the same, optimized grammar,
our recognizer outperforms GLL by a factor of over 300.
Even though our implementation of GLL is most probably suboptimal,
relational parsing remains in clear advantage.

\begin{figure}
  \centering
  \begin{tabular}{l|r|r|r}
    semiring & plain values & substitutions & no evaluation \\
    \hline\hline
    recognition & 7.8 s / 0.22 GB & 79.3 s / 0.54 GB & 25.1 s / 0.35 GB \\
    counting & 62.5 s / 1.43 GB & 82.7 s / 0.57 GB & 25.4 s / 0.37 GB \\
    parsing & --- & 476.4 s / 0.88 GB & 38.2 s / 0.55 GB \\
    \hline
  \end{tabular}
  \caption{Efficiency of relational parsing over different semirings.}
  \label{fig:timings-parsing}
\end{figure}

We now turn to the evaluation of parsing.
As we have not yet implemented a variant of GLL able to work with an arbitrary semiring,
we only report here on the different variants of our method.

The timings are given by Figure~\ref{fig:timings-parsing}.
The comparison only contains variants with dominator-based memoization.
The column named "no evaluation" is for the substitution-based algorithm,
but without the final evaluation of substitutions to actual semiring value.
In other words, it simulates what happens if we allow the results of parsing to be returned in "compressed" form.

Comparing the numbers here with those from Figure~\ref{fig:timings-recognition},
one can see our full parsing algorithm to take the same time as the GLL \emph{recognizer}
(and this only if we choose the "better" variant of the grammar for each of them;
on the same, optimized grammar, the speedup is about 5 times).
Shall the latter need to construct an SPPF,
our method would again be in greater advantage.

Moreover, computing the final value takes more than 90\% of parsing time.
We are almost certain that this is too much,
and that a better representation of substitutions
(or a more clever way of evaluating them)
would make our algorithm even better.

\section{Related work}
\label{sec:related}

Parsing is a vast subject (see~\cite{grune} for a not-so-recent survey, with over 1500 references!),
thus here we only touch on what we find to be most relevant to our proposed algorithm.

\paragraph*{Simulating nondeterministic pushdown automata}

Effective simultaneous simulation of possible runs of a PDA has been proposed by Lang~\cite{pda-simulation}.
In~\cite{tomita,gss}, Tomita presented the graph-structured stack as a convenient structure for capturing such nondeterminism.
This has led to a multitude of generalized parsers, differing primarily by the automaton being simulated:
this includes the GLR family (e.g.,~\cite{brnglr}), GLC~\cite{glc,glc-improved}, and GLL~\cite{gll,gll-performance}.

As Lang's, our method could be adjusted to simulate other versions of pushdown automata.
It is even conceivable that an appropriate "phase function" could be obtained automatically for an arbitrary PDA,
and still require only a constant number of primitive language operations.
We have not investigated this path.

The Adaptive LL(*) algorithm~\cite{allstar}, underlying the popular ANTLR4 parser generator, uses simulation to predict valid transitions.
It represents configuration sets explicitly, and thus fails to handle left-recursive grammars, where these sets can be infinite.
Our approach, originally aimed to remove this restriction, abstracts away the specific representation.
In fact, we have first used a "standard" GSS instead of the proposed DAG structure---%
however, in such case we failed to attain linear complexity on some LR-regular grammars.

\paragraph*{Left- and right-recursive productions}

For top-down (predictive) parsers, left-recursive productions are inherently problematic,
because it is difficult to know up front how many times a \texttt{call} loop should be followed.
Thus, most algorithms in the LL family refuse to handle such grammars.
GLL captures \emph{all} possible runs through left recursion by introducing cycles into its graph-structured stack,
and so does Earley's algorithm (cyclicity here takes the form of an item referencing its own layer).

Right recursion requires special attention when reduction possibilities are considered,
because it may happen that a new way of reaching some state in discovered \emph{after} it has been "processed".
All nondeterminism-simulating parsers handle this using work queues or other mutable data structures.
Similar care (and mutability) is needed in Earley-style parsers.

With simultaneous left and right recursion,
the grammar is infinitely ambiguous,
and thus the final result of parsing \emph{must} be cyclic.
Relational parsing integrates both left and right recursion into the atomic languages,
thus completely avoiding the need for cyclic structures and mutation during parsing.
To our best knowledge it is the first algorithm with this property.

Right recursion influences also the complexity of generalized parsers,
as it can make linearly long sequences of reductions available at each input position.
Some parsers, including Earley's, use lookahead to avoid incurring quadratic complexity on LR(k) grammars.
Another approach, used in Leo's modification~\cite{leo} and applicable to all LR-regular grammars \emph{without} using lookahead,
is to create shortcuts through such long reduction paths.
This is precisely what we achieve by the decomposition proposed in Proposition~\ref{prop:language-shape}.

\paragraph*{Nullables}

Nullable states (nonterminals) seem to be a nuisance for every parsing algorithm.
Although theoretically one could avoid dealing with them by rewriting the grammar,
it is impractical to do so: binary grammars can grow quadratically, and arbitrary ones---exponentially.
Right-nulled GLR~\cite{rnglr,rnglr-optimize} avoids the exponential blow-up by only eliminating trailing nullables in each production.
Most other algorithms attempt to handle all nullables "on the fly".
In our case, this amounts to making the atomic languages null-closed.
While it does not cause the automata to have more states, additional edges \emph{are} introduced and traversed.

In an analogous situation, the version of CYK presented in~\cite{cyk-2nf} does \emph{not} precompute the extra edges.
Instead it avoids higher grammar-related complexity by grouping together multiple walks on its "automata".
It is interesting whether this idea could somehow be adapted to our setting.

\paragraph*{Atomic closures vs.~left corners}

The slightly less popular GLC parser builds a bridge between top-down and bottom-up parsing by using the \bld{left corner relation},
containing a pair of nonterminals whenever one of them can appear leftmost in some derivation from the other.
Our proposed implementation of atomic closures is indexed by pairs of states related in exactly the same way,
and the appropriate rules of GLC mirror the transitions in our atomic automata.

The authors of GLC propose to make the algorithm more effective by grouping some of the cases heuristically.
With our approach, it is clear that one could use any automaton recognizing the same language.
This opens the way for well-understood optimizations.

\paragraph*{Memoization}

Memoization has been put to great use in the Adaptive LL(*) parsing algorithm,
even though it is only used there when predicting the right nondeterministic choice.
To make memoized computations applicable in more situations, the algorithm pretends to forget all but the top state of the RTN.
If such "blind" computation succeeds, it can then be applied on top of any stack.
However, should it fail, the algorithm resorts to non-memoized computation with full stack information.
Our dominator-based memoization handles such situations better, and in a way optimally:
we always look at exactly such stack depth as is necessary.

\paragraph*{Limiting stack activity}

It is well known that the primary source of inefficiency in generalized parsers is handling large graph-like structures (i.e., the GSS), offering poor cache locality.
Reduction-incorporated parsers (cf.~\cite{riglr}) attempt to limit the size of these graphs by handling most of the grammar using finite automata, and using the stack only when necessary.
Appropriate pushdown automata can be constructed in multiple ways,
and thus in~\cite{riglr-optimize} the authors suggest selecting the best one by profiling the parser on sample inputs.

In our algorithm, \emph{all} activity between successive terminals is handled by the automata for atomic languages.
Furthermore, instead of limiting the depth of the (multiple) stacks, we perform most---even locally ambiguous---operations on a \emph{single} stack.
With dominator-based memoization, the algorithm "learns" the right stack actions \emph{while} parsing,
though one could also possibly "train" it beforehand on some sample inputs.

\paragraph*{Semantics}

As far as we know, no previous generalized parsing algorithm \emph{explicitly} computes the languages of reachable configurations.
However, these languages are implicitly present in \emph{all} of them,
and often could be easily read off their data structures after processing each terminal:
in GLL for example, reachable configurations are precisely all the paths through the GSS.

We believe that our semantics-based approach opens the way not only to cleaner and more performant parsers for all context-free languages,
but possibly also to handling larger useful grammar classes,
beginning with conjunctive grammars~\cite{conjunctive} or even grammars with one- or two-sided contexts~\cite{glr-with-contexts,bicontexts,recognize-bicontexts}.

\bibliographystyle{plainurl}
\bibliography{parsing}

\end{document}